\documentclass[sigconf]{acmart}

\usepackage{algorithm}
\usepackage{algpseudocode}
\usepackage{balance} 
\usepackage{color}
\usepackage{enumitem}
\usepackage[T1]{fontenc}
\usepackage{hyperref}
\usepackage{multirow}

\usepackage{graphicx}
\usepackage{graphics}{\tiny}
\graphicspath{{figure/}}

\usepackage[caption=false,font=footnotesize,labelfont=sf,textfont=sf]{subfig}

\usepackage{fancybox}

\usepackage{array}

\usepackage{url}

\newcommand{\Adv}{\mathcal{A}}
\newcommand{\Clg}{\mathcal{C}}
\newcommand{\GC}{\mathcal{GC}}
\newcommand{\SF}{\mathcal{SF}}
\newcommand{\IS}{\mathcal{IS}}
\newcommand{\ISC}{\mathcal{ISC}}
\newcommand{\Leakage}{\mathcal{L}}
\newcommand{\OXT}{\mathrm{OXT}}
\newcommand{\PRF}{\mathrm{PRF}}
\newcommand{\DB}{\mathrm{DB}}
\newcommand{\EDB}{\mathrm{EDB}}
\newcommand{\TSet}{\mathrm{TSet}}
\newcommand{\XSet}{\mathrm{XSet}}
\newcommand{\View}{\mathrm{VIEW}}
\newcommand{\Out}{\mathrm{OUT}}
\newcommand{\Real}{\mathrm{REAL}}
\newcommand{\Ideal}{\mathrm{IDEAL}}
\newcommand{\query}{\mathbf{q}}
\newcommand{\id}{\mathbf{id}}
\newcommand{\stag}{\mathbf{stag}}
\newcommand{\xtag}{\mathbf{xtag}}
\newcommand{\xtoken}{\mathbf{xtoken}}
\newcommand{\Sim}{\mathit{Sim}}
\newcommand{\Client}{\mathit{C}}
\newcommand{\Server}{\mathit{S}}
\newcommand{\Party}{\mathit{P}}

\newcommand{\system}{GraphSE\textsuperscript{2}}

\newtheorem{Definition}{Definition}
\newtheorem{Theorem}{Theorem}

\copyrightyear{2019} 
\acmYear{2019} 
\setcopyright{acmcopyright}
\acmConference[AsiaCCS '19]{ACM Asia Conference on Computer and Communications Security}{July 9--12, 2019}{Auckland, New Zealand}
\acmBooktitle{ACM Asia Conference on Computer and Communications Security (AsiaCCS '19), July 9--12, 2019, Auckland, New Zealand}
\acmPrice{15.00}
\acmDOI{10.1145/3321705.3329803}
\acmISBN{978-1-4503-6752-3/19/07}

\begin{document}
\title{\system: An Encrypted Graph Database for Privacy-Preserving Social Search}

\author{Shangqi Lai}
\additionalaffiliation{
\institution{Data61, CSIRO, Melbourne, Australia}
\country{Australia}}
\affiliation{
\institution{Monash University}
\city{Melbourne}
\country{Australia}}
\email{shangqi.lai@monash.edu}

\author{Xingliang Yuan}
\affiliation{
\institution{Monash University}
\city{Melbourne}
\country{Australia}}
\email{xingliang.yuan@monash.edu}

\author{Shi-Feng Sun}\authornotemark[1]
\affiliation{
\institution{Monash University}
\city{Melbourne}
\country{Australia}}
\email{shifeng.sun@monash.edu}

\author{Joseph K. Liu}\authornote{Corresponding author.}
\affiliation{
\institution{Monash University}
\city{Melbourne}
\country{Australia}}
\email{joseph.liu@monash.edu}

\author{Yuhong Liu}
\affiliation{
\institution{Santa Clara University}
\city{Santa Clara}
\country{U.S.}}
\email{yhliu@scu.edu}

\author{Dongxi Liu}
\affiliation{
\institution{Data61, CSIRO}
\city{Syndey}
\country{Australia}}
\email{dongxi.liu@data61.csiro.au}

\begin{abstract}

In this paper, we propose \system, an encrypted graph database for online social network services to address massive data breaches. \system~preserves the functionality of \emph{social search}, a key enabler for quality social network services, where social search queries are conducted on a large-scale social graph and meanwhile perform set and computational operations on user-generated contents. 
To enable efficient privacy-preserving social search, \system~provides an encrypted structural data model to facilitate parallel and encrypted graph data access. It is also designed to decompose complex social search queries into atomic operations and realise them via interchangeable protocols in a fast and scalable manner.  
We build \system~with various queries supported in the Facebook graph search engine and implement a full-fledged prototype. Extensive evaluations on Azure Cloud demonstrate that \system~is practical for querying a social graph with a million of users.
\end{abstract}

\ccsdesc[500]{Security and privacy~Management and querying of encrypted data}
\ccsdesc[500]{Security and privacy~Social network security and privacy}
\ccsdesc[500]{Security and privacy~Privacy-preserving protocols}
\keywords{Social Search, Graph Database, Encrypted Query Processing} 

\maketitle

\section{Introduction}
Data breaches in online social networks (OSNs) affect billions of individuals and raise critical privacy concerns across the entire society~\cite{LiangLLW15,info2018breach}.  
Besides, driven by the demands on huge storage and computation resources, OSN service providers utilise public commercial clouds as their back-end data storage~\cite{AWS2018Airbnb, AWS2018PIXNET, instagram2011power}, which further broadens the attack plane~\cite{ren2012security}. 
Therefore, there is an urgent call to improve the control of data confidentiality for cloud providers \cite{YangHL16,LiuLSLX16,LiZLQD16}, in particular for current OSN services. 
The prevailing consensus to prevent data leakage is encryption. 
However, this approach impairs the functionality of {\emph social search}, a key enabler for quality OSN services~\cite{dan2012social}.
Social search allows users to search content of interests created by their friends. Compared with traditional web search, it produces personalised search results and serves for a wide range of OSN services such as friend discovering and user targeting. 

The first task to enable privacy-preserving social search is how to scalably query over very large encrypted social graphs. 
On the one hand, a typical social graph can contain millions or even billions of users. 
On the other hand, users may generate large volume of contents which will be queried for social search related services~\cite{curtiss2013unicorn}. 
The second and more challenging task is how to realise complex social search queries in an efficient and secure manner. 
As developed in plaintext systems (e.g., Facebook's Unicorn~\cite{curtiss2013unicorn}), queries of social search contains set operations on graph-structured data, and the retrieved contents from the graph need to further be analysed (e.g., aggregation and sorting) for advanced services such as friendship-based recommendation. 

In the literature, some work~\cite{nayak2015graphsc,blanton2013data} leverages generic building blocks (e.g., garbled circuits and oblivious data structures) to devise secure computational frameworks for graph algorithms. However, those frameworks do not appear to be scalable for low latency queries over large graphs. For example, a recent garbled circuits based framework~\cite{nayak2015graphsc} takes several minutes to complete a sorting algorithm over a graph with only tens of thousands of nodes. 
Other work focuses on dedicated privacy-preserving graph algorithms, e.g., neighbour search~\cite{chase2010structured, kamara2018structured}, and shortest distance queries~\cite{xie2016practical,meng2015grecs,wu2016privacy,wang2017secgdb}. Unfortunately, the above algorithms are limited for or different from the functionality of social search queries.

\noindent{\bf Contributions.}
To bridge the gap, in this paper, we propose and implement \system, the first encrypted graph database that supports privacy-preserving social search. 
Unlike prior work which either suffers from low scalability or limited functionality, \system~enables scalable queries over very large encrypted social graphs, and preserves the rich functionality of the plaintext social search systems. Our contributions can be summarised as follows:

\begin{itemize}[leftmargin=*]

\item We propose an encrypted and distributed graph model built on social graph modelling, searchable encryption, and the data partition technique. 
It facilitates queries over encrypted graph partitions in parallel, and maintains the locality of graph data and user-generated contents for low query latency. 



\item We devise mixed yet interchangeable protocols to enable complex social search functions. The way of doing this is to decompose queries into atomic operations (i.e., set, arithmetic, and sorting operations) and then adapt suitable cryptographic primitives for efficient realisation. All these operations are tailored to be executed in parallel.

\item We realise query operators of the Facebook's social search system Unicorn~\cite{curtiss2013unicorn}, i.e., \textbf{term}, \textbf{and}, \textbf{or}, \textbf{difference}, and \textbf{apply}. We also design a query planner that can parse a query to atomic operations and initiate the corresponding primitives. 

\item We formally prove the security of our proposed query protocols under the real and ideal paradigm. Queries, graph data, and results are protected throughout the query process.  

\item We show the practicality of \system~by implementing a prototype which is readily deployable. It leverages Spark~\cite{zaharia2010spark} for setup (data partition and encryption), Redis~\cite{redis} as the storage back-end, and uses Apache Thrift~\cite{slee2007thrift} to implement the query planner and query processing logic.  

\end{itemize}

Our comprehensive evaluation on the Youtube dataset~\cite{mislove-2007-socialnetworks} with 1 million nodes confirms that all atomic operations are of practical performance. For set queries, \system~retrieves a content list with $500$ entities within $10$ ms. For an average user ($130$ friends), \system~takes at most $20$ ms if the set operation involves two indexing terms (attributes); and it takes no more than $100$ ms for five indexing terms. 
Regarding the computational operations, \system~takes $100$ ms to handle arithmetic computations over $10^4$ entities, and $450$ ms to sort $128$ entities. 
As a summary, most of the queries for an average user are processed within $1$ s, and throughput is reduced at most $49\%$ compared to the plaintext queries.

\noindent{\bf Organisation.}
The rest of this paper is structured as follows. We discuss related work in Section~\ref{sec:related}. After that, we describe the system overview in Section~\ref{sec:overview}, and present the encrypted and distributed graph data model and the design of atomic operations in Section~\ref{sec:modelling}. In Section~\ref{sec:queries}, we introduce the realisation of privacy-preserving social search queries and their security. 
Next, we describe our prototype implementation in Section~\ref{sec:impl}, and evaluate the performance in Section~\ref{sec:evaluation}. We give a conclusion in Section~\ref{sec:conclusion}.

\section{Related Work}\label{sec:related}
{\bf Privacy-preserving graph query processing.} There exist various designs that aim to answer a certain type of queries over the encrypted graph. 
%
%
Structured encryption~\cite{chase2010structured} is proposed in the framework of SSE and supports adjacency and neighbouring queries.
%
Some recent work is proposed to support privacy-preserving subgraph queries~\cite{cao2011privacy, chang2016privacy}. 
%
%
%
%
However, all the above designs enable limited query functionality.
%
%
Another line of work on privacy-preserving graph processing is to perform shortest-path queries over the encrypted graph.
Protocols for this type of queries are devised via oblivious RAM~\cite{xie2016practical}, structured encryption~\cite{meng2015grecs}, or Garbled Circuit~\cite{wu2016privacy,wang2017secgdb}. 
%
%
To implement more complicated algorithms, protocols are proposed to use secret sharing and homomorphic encryption for Breadth-first search (BFS)~\cite{blanton2013data}, PageRank~\cite{xie2014cryptgraph}, and approximate eigen-decomposition~\cite{sharma2018privategraph}. 
We stress that the above work targets on different query functionality other than social search. 
Note that a recent framework named GarphSC~\cite{nayak2015graphsc} can generate data-oblivious Garbled Circuit (GC) for graph algorithms such as PageRank and Matrix Factorisation. 
Because oblivious data structures are adapted for large graphs and all computations are realised via GC, it does not appear to achieve low latency for social search queries. 



\noindent{\bf Encrypted database system.} Our system is also related to encrypted database systems~\cite{popa2011cryptdb,pappas2014blind,poddar2016arx,papadimitriou2016big,yuan2017enckv}.
%
CryptDB~\cite{popa2011cryptdb} is the first practical encrypted database system, which is built on property-preserving encryption (PPE). It supports SQL queries over encrypted relational data records. 
BlindSeer~\cite{pappas2014blind} proposes a Bloom Filter based index and leverages GC to evaluate arbitrary boolean queries with keywords and ranges. 
Arx~\cite{poddar2016arx} follows the design of CryptDB to support SQL queries, but it uses SSE and GC to reduce the leakage from PPE.
Seabed~\cite{papadimitriou2016big} uses additively symmetric homomorphic encryption (ASHE) to perform efficient aggregation over the encrypted data, and develops a schema with padding to mitigate the inference attack~\cite{naveed2015inference}. 
EncKV~\cite{yuan2017enckv} adapts SSE and ORE schemes to design an encrypted and distributed key-value store. 
However, all the encrypted databases mentioned above are neither designed for graph data nor optimised for social search. 

\noindent{\bf Graph processing system.} In the plaintext domain, a large number of graph processing systems~\cite{low2012distributed,curtiss2013unicorn,chi2016nxgraph} (just to list a few) are proposed to support efficient large graph processing. 
%
%
However, all the above systems only support queries over the graphs in unencrypted form, which are unable to address privacy concerns of sensitive data leakage. 
Authenticated graph query~\cite{goodrich2011efficient} is proposed to verify the correctness of graph queries, which could be a complementary work to prevent attacks from malicious adversaries. 

\section{Background} \label{sec:pre}
\subsection{Social Graph Model}\label{subsec:social} 
The social graph consists of nodes (aka entities) and edges (aka relationships of entities) in social networks.
%
%
As the social graph is a sparse graph~\cite{curtiss2013unicorn}, it is normally represented via a set of adjacency lists. Like~\cite{curtiss2013unicorn}, we refer to these adjacency lists as {\it posting lists}.

Formally, the social graph is an edge-labeled and directed graph $G = (\mathrm{V}, \mathrm{E})$, where $\mathrm{V} = \{\mathrm{v}_1, \mathrm{v}_2, ...\}$ is the entity set and $\mathrm{E} = \{\mathrm{e}_1, \mathrm{e}_2, ...\}$ is the relationship set.
Each posting list contains a list of entities $\{\mathbf{v}\}$, which are (sort-key, $\id$) pairs. The sort-key is an integer that indicates the importance of the entity in a posting list, and the $\id$ is its unique identifier.

The posting lists are indexed by the inverted index, and modelled by the edges in social graph: All edges in $G$ can be represented as a triad $\mathrm{e} = (\mathsf{u}, \mathsf{v}, \textbf{edge-type})$ which consists of its egress, ingress nodes ($\mathsf{u}, \mathsf{v}\in \mathrm{V}$) plus an \textbf{edge-type} which is a string representing the relationship between nodes (e.g., {\bf friend}, {\bf like}).
The inverted indexing term $t$ is in the form of
$\textbf{edge-type}\textbf{:}\id_{\mathsf{u}}$. For example, the user may use $\textbf{friend:}{\id}_i$ to get the posting list of user $i$'s friends.

%

\subsection{Oblivious Cross-Tags ($\OXT$) Protocol}\label{subsec:oxt}
Oblivious Cross-Tags ($\OXT$) Protocol~\cite{cash2013highly} is an SSE protocol, which proceeds between client $\Client$ and server $\Server$. It provides an efficient way to perform conjunctive queries in encrypted database\footnote{The scheme proposed in \cite{kamara2017boolean} supports disjunctive queries, but it consumes large storage space.}. Here we provide a high-level description as needed for the basic operations of our proposed system.


The protocol has two types of data structures. 
Firstly, for every keyword $w$, an inverted index, referred as `$\TSet(w)$', is built to point to the set $\DB(w)$ of all entity identifiers $\id$s associating with $w$. Each $\TSet(w)$ is identified by an indexing term called $\stag(w)$, and all $\id$ values in $\TSet(w)$ are encrypted via a secret key $K_w$. Both $\stag(w)$ and $K_w$ are computed as a $\PRF$ applied to $w$ with $\Client$'s secret keys.
Another data structure called `$\XSet$' is built to hold a list of hash values $h(\id,w)$ (called '$\xtag$') over all entity identities $\id$ and keywords $w$ contained in $\id$, where $h$ is a certain (public) cryptographic hash function. The above two data structures are stored on the server-side.

To search a conjunctive query $(w_1, w_2, ..., w_n)$ with $n$ keywords, $\Client$ sends the `search token' $\stag(w_1)$ related to $w_1$ (called `s-term', we assume it to be $w_1$ in the above query) to $\Server$, which allows the server to run $\mathrm{TSet.Retrieve}(\TSet, \stag(w_1))$ and retrieve $\TSet(w_1)$ from the $\TSet$. 
In addition, $\Client$ sends `intersection tokens' $\xtoken(w_1, w_i)$ (called `xtraps') related to the $n-1$ keyword pairs $(w_1,w_i)$ consisting of the `s-term' paired with each of the remaining query keywords $w_i$, $2\leq i \leq n$ (called `x-terms'). 
The xtraps allow the server to evaluate the cryptographic hash function of pairs $(\id,w_i)$ without knowing either keyword $w_i$ or $\id$. $\mathrm{S}$ checks the existence of $h(\id,w_i)$ in $\XSet$ and filters the $\TSet(w_1)$ to $n-1$ subsets of entities that contain the pairs $(w_1,w_i)$. 
It only returns the entities that contain all $\{w_i\}_{1\leq i\leq n}$ to the client. 
$\Client$ finally uses $K_w$ to recover the $\id$s of entities.

As mentioned in~~\cite{cash2013highly}, the security of $\OXT$ parameterised by a leakage function $\Leakage_{OXT}=(N, \mathrm{\phi}, \bar{s}, \mathrm{SP}, \mathrm{XP}, \mathrm{RP}, \mathrm{IP})$.
It depicts what an adversary is allowed to learn about the database and queries via executing $\OXT$ protocol. 
Informally, considering a vector of queries $\mathsf{q}=(\mathsf{s}\wedge \mathsf{\phi}(\mathsf{x}_2, ..., \mathsf{x}_n))$, which consists of a vector of s-terms $\mathsf{s}$, a vector of boolean formulas $\mathsf{\phi}$, and a sequence of x-term vectors $\mathsf{x}_2, ..., \mathsf{x}_n$.
After executing $\mathsf{q}$ in a chosen database $\DB$, the adversary only can learn:
\begin{itemize}[leftmargin=*]
\setlength{\itemsep}{0pt}
\setlength{\parsep}{0pt}
\setlength{\parskip}{0pt}
\item $N$: The total number of $(\id,w)$ pairs.
\item $\mathsf{\phi}$: The boolean formulae that the client wishes to query.
\item $\bar{s}$: The repeat pattern in $\mathsf{s}$.
\item $\mathrm{SP}$: The size of posting lists for $\mathsf{s}$.
\item $\mathrm{XP}$: The number of x-terms for each query.
\item $\mathrm{RP}$: The set of result $\id$ matching each pair of (s-term, x-term)-conjunction which is in the form $(\mathsf{s}, \mathsf{x}_i), 2\leq i \leq n$.
\item $\mathrm{IP}$: The set of result $\id$ both existing in the posting lists of $\mathsf{s}[i]$ and $\mathsf{s}[j]$, which is only revealed when two queries $\mathsf{q}[i], \mathsf{q}[j], i\neq j$ have different s-terms but same x-terms.
\end{itemize}
%

\subsection{Secure Computation}\label{subsec:secureComp}
{\bf Additive Sharing and Multiplication Triplets.} To additively share ($Shr^A(\cdot)$) an $\ell$-bit value $a$, the first party $\Party_0$ generates $a_0\in \mathbb{Z}_{2^l}$ uniformly at random and sends $a_1=a-a_0 \mod 2^l$ to the second party $\Party_1$. 
The first party's share is denoted by $\langle a\rangle_0^A=a_0$ and the second party's is $\langle a\rangle_1^A=a_1$, the modulo operation is omitted in the description later. 
To reconstruct ($Rec^A(\cdot,\cdot)$) an additively shared value $\langle a\rangle^A$ in $\Party_i$, $\Party_{1-i}$ sends $\langle a\rangle_i^A$ to $\Party_i$ who computes $\langle a\rangle_0^A+\langle a\rangle_1^A$.
Given two shared values $\langle a\rangle^A$ and $\langle b\rangle^A$, Addition ($Add^A(\cdot,\cdot)$) is easily performed non-interactively. 
In detail, $\Party_i$ locally computes $\langle c\rangle_i^A=\langle a\rangle_i^A+\langle b\rangle_i^A$, which also can be denoted by $\langle c\rangle^A=\langle a\rangle^A + \langle b\rangle^A$. 
To multiply ($Mul^A(\cdot,\cdot)$) two shared values $\langle a\rangle^A$ and $\langle b\rangle^A$, we leverage Beaver's multiplication triplets technique~\cite{beaver1991efficient}. 
Assuming that the two parties have already precomputed and shared $\langle x\rangle^A$, $\langle y\rangle^A$ and $\langle z\rangle^A$, where $x,y$ are uniformly random values in $\mathbb{Z}_{2^l}$, and $z=xy\mod 2^l$. 
Then, $\Party_i$ computes $\langle e\rangle_i^A=\langle a\rangle_i^A-\langle x\rangle_i^A$ and $\langle f\rangle_i^A=\langle b\rangle_i^A-\langle y\rangle_i^A$. 
Both parties run $Rec^A(\langle e\rangle_0^A, \langle e\rangle_1^A)$ and $Rec^A(\langle f\rangle_0^A, \langle f\rangle_1^A)$ to get $e,f$, and $\Party_i$ lets $\langle c\rangle_i^A=i\cdot e\cdot f+f\cdot\langle x\rangle_i^A+e\cdot \langle y\rangle_i^A+\langle z\rangle_i^A$.

\noindent{\bf Garbled Circuit and Yao's Sharing} Yao's Garbled Circuit (GC) is first introduced in~\cite{yao1982protocols}, and its security model has been formalised in~\cite{bellare2012foundations}. 
GC is a generic tool to support secure two-party computation. 
The protocol is run between a ``garbler'' with a private input $x$ and an ``evaluator'' with its private input $y$. The above two parties wish to securely evaluate a function $f(x,y)$. At the end of the protocol, both parties learn the value of $z=f(x,y)$ but no party learns more than what is revealed from this output value. 
In details, the garbler runs a garbling algorithm $\GC$ to generate a garbled circuit $F$ and a decoding table $dec$ for function $f$. The garbler also encodes its input $x$ to $\hat{x}$ and sends it to the evaluator. The evaluator runs an oblivious transfer (OT)~\cite{asharov2013more} protocol with the garbler to acquire its encoded input $\hat{y}$. Finally, the evaluator can compute $\hat{z}$ from $F, \hat{x}, \hat{y}$, decode it with $dec$, and share the result $z$ with the garbler. 
The security proof against a semi-honest adversary under two-party setting is given in~\cite{lindell2009proof}.

In the following parts, we assume that $\Party_0$ is the garbler and $\Party_1$ is the evaluator.
GC can be considered as a protocol which takes as inputs the Yao's shares and produces the Yao's shares of outputs. 
In particular, the Yao's shares of 1-bit value $a$ is denoted as $\langle a\rangle_0^Y=\{K_0, K_1\}$ and $\langle a\rangle_1^Y=K_a$, where $K_0, K_1$ are the labels representing $0$ and $1$, respectively. 
The evaluator uses its shares to evaluate the circuit and gets the output shares (another labels).

Additive shares can be switched to Yao's shares efficiently. 
To be more precise, two parties secretly share their additive shares $a_0=\langle a\rangle_0^A$, $a_1=\langle a\rangle_1^A$ in bitwise via Yao's sharing. 
The evaluator then receives $\langle a_0\rangle^Y$ and $\langle a_1\rangle^Y$ and evaluates the circuit $\langle a_0\rangle^Y+\langle a_1\rangle^Y \mod \langle2^l\rangle^Y$ to get the label of $a$.

\section{System Overview}\label{sec:overview}
\begin{figure}[!t]
\centering
\includegraphics[width=\linewidth]{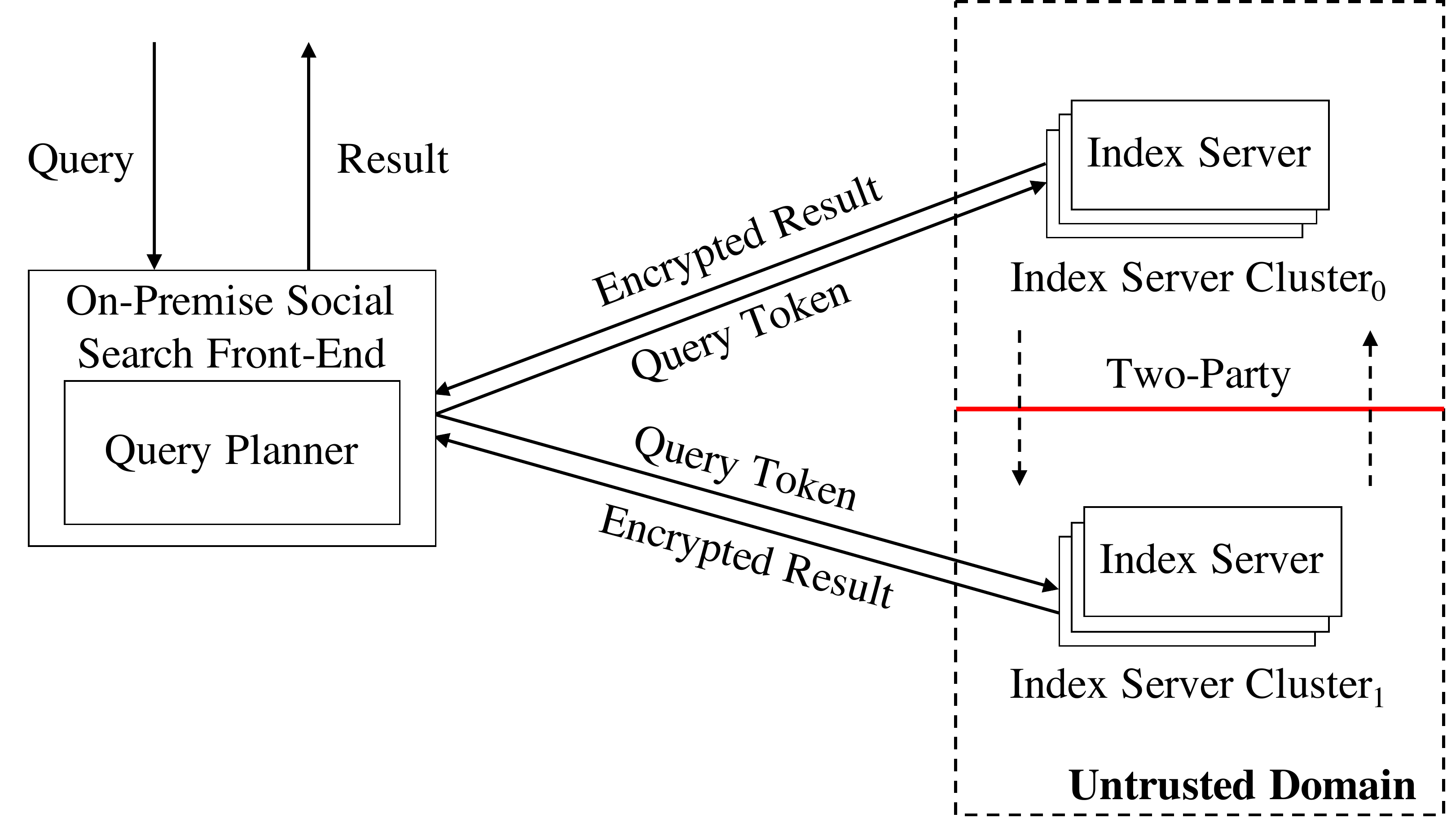}
\caption{System architecture overview.}
\vspace{-10pt}
\label{fig:overview}
\end{figure}
\subsection{System Architecture}\label{subsec:arch}
As shown in Figure~\ref{fig:overview}, \system~has two entities: the on-premise social search service front-end ($\SF$) and the index server cluster ($\ISC$) with several index servers ($\IS$s) in an untrusted cloud. 
Note that this setting is consistent with many off-the-shelf social network service providers such as Airbnb~\cite{AWS2018Airbnb} and Instagram~\cite{instagram2011power}, who use cloud data storage as the back-end to manage large graphs and massive user-generated data contents. Also, such architecture is now natively supported by public clouds, e.g., AWS Outposts~\cite{AWS2018outposts}. \system~aims to improve the protection of data confidentiality at the back-end, which is usually the high-value target for adversaries in practice. 

During the setup phase, $\SF$ partitions the social graph to disjoint subgraphs and builds two instances of SSE indexes of each subgraph for the queries on structured information. The generated indexes are uploaded to two non-colluded $\ISC$s with multiple $\IS$s respectively. 
The sort-keys are co-located with the corresponding indexes in the form of additive shares on the above two $\ISC$s for the arithmetic operations and sorting.
Specifically, each $\IS$ has one of the two additive shares, and it pairs with a counter-party in the other cluster, which maintains the same index but holds the other share.
Upon receiving a query from its users, $\SF$ uses a query planner to parse the query into atomic operations (see Section~\ref{subsec:atomic}) to generate a query plan. It then sends the query tokens of atomic operations to all $\IS$s to execute the query plan. After that, each $\IS$ requests the structured information via the tokens. Based on the matched encrypted contents, it executes arithmetic operations and scoring/ranking algorithms with its counter-party. Finally, the encrypted result is returned to $\SF$. 


%
%
In this architecture, we consider a scenario of secure computation sourcing where the in-house $\SF$ assigns the computation to the $\IS$s in two untrusted but non-colluding clusters $\ISC_0$ and $\ISC_1$. 
Such a model of secure multi-party computation is formalised in~\cite{kamara2011outsourcing} and applied in many existing studies~\cite{nikolaenko2013privacy, baldimtsi2015sorting,mohassel2017secureml}. 
Built on this model, \system~offers two advantages: (i) $\SF$ is not required to be involved with any computation after it distributes the data to the servers, and (ii) the computation process can benefit from the mixture of multi-party computation protocols that enable efficient arithmetic operation, comparison, and sorting at the same time.
Note that the communication between $\IS$s will not be the system bottleneck, because $\IS$s can be deployed in cloud clusters with dedicated datacenter networking support. 
This is consistent with prior studies based on the same architecture~\cite{mohassel2017secureml}.

\begin{table*}[!t]
	\centering
	\caption{Supported social search operators in \system~and its essential atom operations.}
	\label{tlb:atomic}
	\begin{tabular}{|c|c|c|c|c|c|}
	\hline
	\multirow{2}{*}{Query operator}&\multirow{2}{*}{Example(from~\cite{curtiss2013unicorn})} & \multicolumn{4}{ c| }{Atomic operations} \\
	\cline{3-6}
	& & \textit{Index Access} & \textit{Set Operations} & \textit{Arithmetic} & \textit{Sorting} \\
	\hline
	\textbf{term} & (\textbf{term} \textbf{friend:1}) & $\surd$ & & & $\surd$\\
	\hline 
	\textbf{and} & (\textbf{and} \textbf{friend:1} \textbf{friend:2}) & $\surd$ & $\surd$ & & $\surd$\\
	\hline 
	\textbf{or} & (\textbf{or} \textbf{friend:1} \textbf{friend:2}) & $\surd$ & $\surd$ & & $\surd$\\
	\hline 
	\textbf{difference} & (\textbf{difference} \textbf{friend:3} (\textbf{and} \textbf{friend:1} \textbf{friend:2})) & $\surd$ & $\surd$ & & $\surd$\\
	\hline
	\textbf{apply} & (\textbf{apply friend:} \textbf{friend:1}) & $\surd$ & $\surd$ &$\surd$ & $\surd$\\
	\hline
	\end{tabular}
	\vspace{-10pt}
\end{table*}

\subsection{High-level Description}
%
Before introducing the details of our system, we elaborate on the design overview and underlying design intuitions. 
To query large social graphs, \system~develops an encrypted and distributed graph model. It is built on graph modelling, searchable encryption, and the standard data partition algorithm. Each server evenly stores an encrypted disjoint part of the whole graph. Meanwhile, this model is designed to co-locate the encrypted contents with the disjoint part containing the users who generate or relate to the contents. As a result, \system~not only maximises the system scalability but also preserves data locality for low query latency.


To facilitate the realisation of various social search queries in the encrypted domain, \system~first splits these complex queries into two stages, i.e., content search over the structured social graph and computational operations on the retrieved contents. Within the above stages, queries are further decomposed into atomic operations, i.e., {\it Index Access}, {\it Set Operations}, {\it Arithmetic} operations, and {\it Sorting}. Since the first stage commonly performs set operations over the social graph, \system~realises our proposed graph model via a well-known searchable encryption scheme for boolean queries (aka OXT~\cite{cash2013highly}). 
%
%
%
%
The second stage requires a combination of different computations to further analyse user contents. For example, collaborative filtering~\cite{breese1998empirical} first obtains the scores of user contents via several addition and multiplication operations and then sorts the scores for an accurate recommendation. 

To accelerate sophisticated computations in the second stage, \system~mixes different secure computation protocols. Note that such philosophy also appears in recent privacy-preserving computation applications~\cite{demmler2015aby, mohassel2017secureml}. 
Unlike prior work, \system~customises the mixed protocols for social search queries and adapts them to our distributed graph model. 
In particular, \system~represents the importance (score) of user-generated contents as the additive shares and deploys two distributed $\OXT$ instances at two non-colluded server clusters to store both the graph partitions and corresponding shares respectively. 
Doing so allows \system~to support parallel and batch addition and multiplication without the interaction between servers\footnote{Multiplication involves a round of interaction between two servers, but they are in the same partition of two clusters.}. 
%
To achieve fast sorting, \system~first converts additive shares to Yao's shares inside garbled circuits (GC) and then invokes a tailored distributed sorting protocol via GC. Each pair of servers in two clusters can perform local sorting in parallel, and then the intermediate results are aggregated for global sorting. Within the protocol, the underlying scores are hidden against servers from either of the two parties. 

\subsection{Threat Assumptions}


In this work, we assume that $\SF$ is a private server dedicatedly maintained by the OSN service provider. It is a trustworthy party in the proposed model. 
Similar to the real-world OSN service provider (e.g. Airbnb), all users should submit their queries to $\SF$ through webpages or mobile apps.
We assume that $\SF$ utilises the secure channel and cryptographic techniques to protect users' secrets.
On the other hand, we assume that all $\IS$s are located in the untrusted domain. 
%
Meanwhile, we consider that the two clusters are semi-honest but not colluding parties. Each cluster performs social search faithfully but intends to learn additional information such as query terms, result $\id$s and ranking values from the graph. 
Besides, those clusters hold user data and perform query functions, and thus they are high-value targets of adversaries. We assume that the two clusters can be compromised by two different passive adversaries, but the two adversaries will not collude.  
\system~aims to protect the confidentiality of the private information in the social graph when the data storage back-end of the social search service is deployed at an untrusted domain.


\subsection{Query Operators}\label{sec:operators}
\system~follows a typical plaintext social search system~\cite{curtiss2013unicorn} to define the operators (see Table~\ref{tlb:atomic}).

In general, all operators in \system~aim to retrieve posting lists from the encrypted graph index. 
The simplest form of these operators is \textbf{term}, which retrieves a single posting list via an {\it Index Access} operation.
Like the other social search system, \system~also supports \textbf{and} and \textbf{or} operators, which yield the intersection and union of posting lists via {\it Set Operations} respectively. In addition, it supports \textbf{difference} operator, which yields results from the first posting list that are not present in the others.
Moreover, \system~supports the unique query operator of Unicorn system~\cite{curtiss2013unicorn}, i.e., \textbf{apply}. The operator allows \system~to perform multiple rounds of posting list retrieval to retrieve contents that are more than one edge away from the source node. 

To enable quality search services (e.g., friendship-based recommendation), the retrieved posting lists should be scored/ranked before returning to users.
As mentioned in Section~\ref{sec:overview}, the additive shares of sort-keys are stored with its indexes. As a result, most of the query operators (e.g., \textbf{term}, \textbf{and}, \textbf{difference} and \textbf{or}) can use these shares to perform {\it Sorting} on the retrieved contents.
Furthermore, it is often useful to return results in an order different from sorting by sort-keys. 
For instance, collaborative filtering~\cite{breese1998empirical} evaluates an arithmetic formula about friendships and ratings on items to produce the personalised scores for recommended items. The new score is a better prediction than the sort-keys, as the later only reflects the overall preference in the community (e.g., the hit-count on the item).
The defined operators natively support arithmetic computations via the additive shares affixed with indexes.
Specifically, \textbf{apply} operator has the capability to support the secure evaluation on complicated scoring formulas with {\it Arithmetic} operations: It can access different types of entities (e.g., user's friends, items liked by users, etc.) in a multiple round-trip query, which means it can combine the scores of different entities and cache the intermediate result for next round computations.

\begin{table}[!t]
	\centering
	\caption{Notations and Terminologies}
	\label{tb:notation}
	\tabcolsep 0.01in
	\begin{tabular}{|c|l|}
		\hline
		Notation & Meaning\\
		\hline
		\hline
		$\id$ & the unique identifier of entity\\
		\hline
		$e_{\id}$ & the encrypted entity $\id$\\
		\hline
		$t$ & an indexing term in the form of $\textbf{edge-type}\textbf{:}\id_{\mathsf{u}}$\\
		\hline
		$\DB$ & an inverted indexed database $\{(t, \{(\text{sort-key},\id)\})\}$ \\
		\hline
		$\DB(t)$ & a list of $\{(\text{sort-key},\id)\}$ indexed by $t$ \\
		\hline
		$\{\mathbf{E}\}$ & the encrypted posting list with $(\langle\text{sort-key}\rangle^A, e_{\id})$ pairs \\
		\hline
		$\Party_i$ & the $i$-th party in \system~($i\in\{0,1\}$)\\
		\hline
		$x$ & a numerical value\\
		\hline
		$\mathbf{X}$ & a matrix\\
		\hline
		$\langle x\rangle^{*}_i$ & the Additive/Yao's share of a numerical value $x$ in $\Party_i$\\
		\hline
		$\langle \mathbf{X}\rangle^{*}_i$ & the Additive/Yao's share of a matrix $\mathbf{X}$ in $\Party_i$\\
		\hline
		$\GC$ & a garbling scheme \\
		\hline
		\hline
	\end{tabular}  
	\vspace{-10pt}
\end{table}

\section{The Proposed System}\label{sec:modelling}
We give a list of needed notations in our system construction and security analysis in Table~\ref{tb:notation}. The detailed definitions of preliminaries we used are given in Section~\ref{sec:pre}.

\subsection{Encrypted Graph Data Model}
To support social search operations in~\cite{curtiss2013unicorn} on an encrypted social graph (see Section~\ref{subsec:social} for details), \system~creates the OXT index (i.e., $\TSet$ and $\XSet$, see Section~\ref{subsec:oxt} for details) for encrypted graph structure access in $\ISC$s, and the additive shares are integrated with the corresponding index to support complex computations.
Specifically, to support simple graph structure data access, each posting list is encrypted and stored as a $\TSet$ tuple in the $\ISC$:
$(\stag(\textbf{edge-type}\textbf{:}\id_{\mathsf{u}}),\{\mathbf{E}\})$.
The $\TSet$ tuple consists of the $\stag$ of indexing term as the key and the encrypted posting list $\{\mathbf{E}\}$ as the value. 
Each element in $\{\mathbf{E}\}$ is an encrypted tuple $\mathbf{E}_{\id}=(\langle\text{sort-key}\rangle^A, e_{\id})$, which keeps the encryption $e_{\id}$ of entity $\id$. 
Additionally, the sort-key of entity is shared as additive sharing value. \system~associates it with the encrypted entity $\id$ to support complex computations.
Moreover, \system~evaluates the cryptographic hash function of $(\id, \textbf{edge-type}\textbf{:}\id_{\mathsf{u}})$ pairs to generate an $\XSet$ for complex set operations.

\subsection{Encrypted and Distributed Graph Index}
In order to support the system to process the query in parallel, \system~distributes the encrypted graph across multiple index servers for each cluster.
\begin{figure}[!t]
\includegraphics[width=\linewidth]{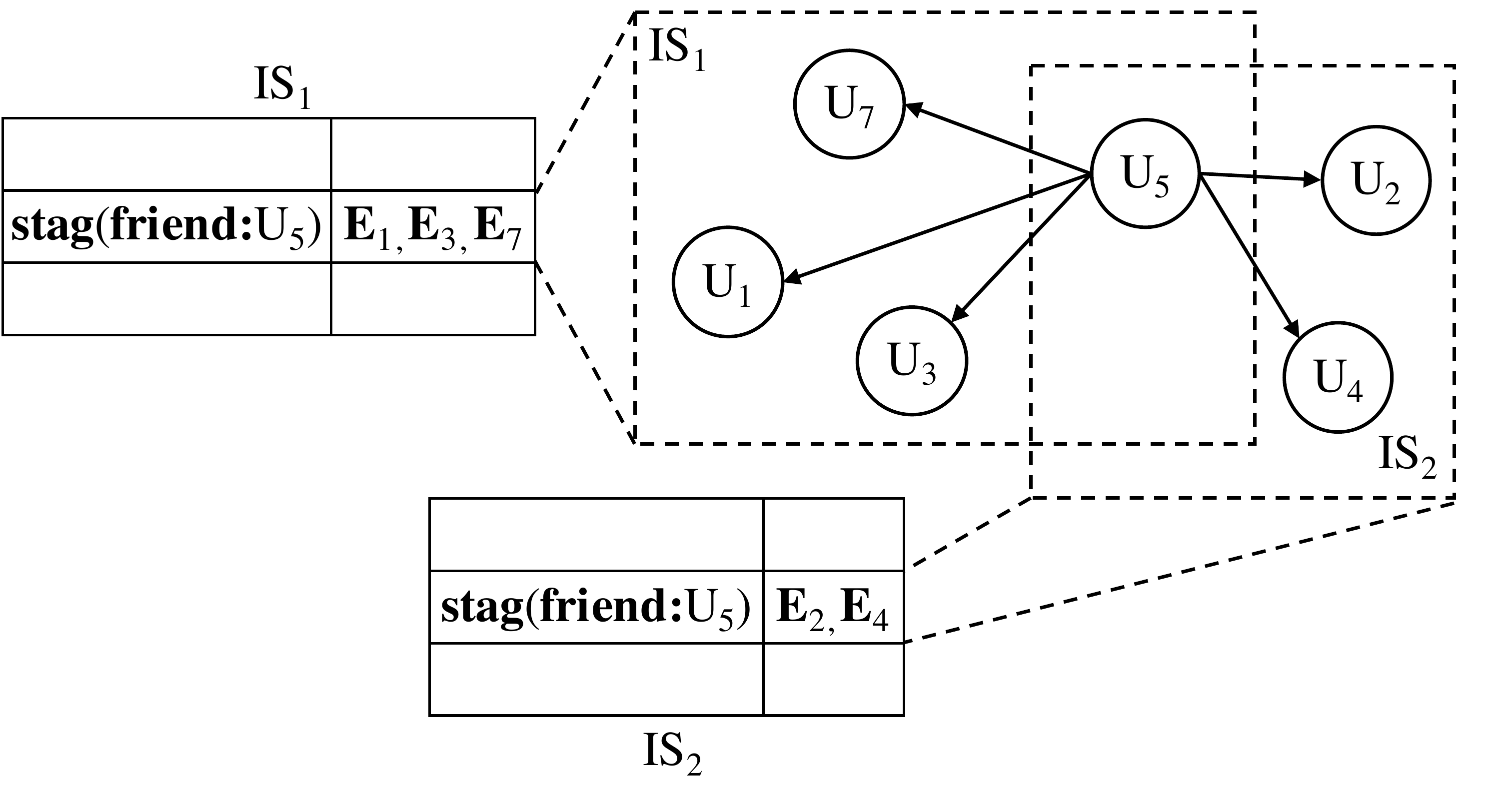}
\caption{Our encrypted and distributed data model, the arrows indicate the friend relationships between users.}
\vspace{-10pt}
\label{fig:dataModel}
\end{figure}
\system~devises a partition strategy that shards the posting lists by hashing on result $\id$.
Figure~\ref{fig:dataModel} gives an example of the proposed partition strategy in an $\ISC$ with two $\IS$s. 
We employ a {\bf modulo} partition strategy, which split the original posting list into multiple non-duplicate parts, but other graph partition strategies (e.g.,~\cite{low2012distributed}) can also be applied to shard the social graph.
The design has three advantages in the context of distributed environment. First, it maintains the availability in the event of server failure.
Furthermore, the sharding strategy enables the distributed system to finish most of the set operations and the consequent scoring, ranking and truncating in $\IS$s. It splits the computation loads into distributed servers to improve the efficiency and also cuts down the communication cost between $\IS$s and $\SF$. 
Finally, it does not affect the security of \system~because the adversary who compromises an $\ISC$ gets the same view (the whole encrypted database) as the adversary in a single $\OXT$ instance. 
If the adversary cannot access all $\IS$s in the $\ISC$, only the view on a fraction of the encrypted database is learned..

\subsection{Atomic Operations}\label{subsec:atomic}
As mentioned in Section~\ref{sec:operators}, the social search queries are implemented by a set of operators. 
We observe that these operators can be decomposed to a set of atomic operations. 
We now describe the implementation of these atomic operations in the encrypted domain. 
%
%
For each atomic operation, we explain how we adapt and optimise it in the proposed system. 

\subsubsection{Index Access} We start with {\it Index Access} operation, which is used to retrieve the neighbouring nodes of the target user with the given \textbf{edge-type} (e.g., {\bf friend}, {\bf likes}) from the social graph. 
Algorithm~\ref{alg:indexAccess} outlines the searching procedure using $\TSet$ operations. 
On receiving the search keyword, $\SF$ firstly generates a search token $\tau$, which is $\stag(t)$ of the indexing term $t$.
$\IS$ can use $\tau$ to search $\TSet$ and get the encrypted posting list $\TSet(t)$ as the return. 
{\it Index Access} operation can be easily extended to run in parallel. More specifically, $\SF$ broadcasts search token $\tau$ to all $\IS$s. After that, each $\IS$ uses $\tau$ to get its local partition of the whole encrypted posting list and sends it back.

\noindent{\bf Security.} The security of {\it Index Access} is guaranteed by the security property of $\TSet$. 
Informally, {\it Index Access} is $\Leakage_{\bf T}$-semantically-secure against adaptive attacks where $\Leakage_{\bf T}$ is the leakage function of $\TSet$. $\Leakage_{\bf T}$ is well-defined and discussed in~\cite{cash2013highly}. 
It ensures {\it Index Access} only leaks the number of edges in the encrypted social graph.

\begin{algorithm}[!t]
\caption{Index Access}
\label{alg:indexAccess}
\begin{algorithmic}[1]
\Require $\TSet$, Indexing Term $t$
\Ensure Encrypted Result $\TSet(t)$
\Function{IndexAccess}{$\TSet, t$}
	\State $\SF$ inputs indexing term $t$, and $\IS$ inputs $\TSet$;
	\State $\SF$ computes $\tau\leftarrow\stag(t)$;
	\State $\SF$ sends $\tau$ to $\IS$;
	\State $\IS$ computes $\TSet(t)\leftarrow\mathrm{TSet.Retrieve}(\TSet, \tau)$;
	\State\Return $\TSet(t)$;
\EndFunction
\end{algorithmic}
\end{algorithm}
\subsubsection{Set Operations} This operation involves the boolean expression with multiple indexing terms. \system~uses it to query the encrypted graph-structured data and finds the neighbouring nodes and the corresponding user-generated content that satisfy the given boolean expression. 
In \system, we adapt $\OXT$ protocol to support this atomic operation, but some of the other SSE protocols supporting conjunctive queries (e.g.~\cite{lai2018result}) can also be readily adapted as the building block of \system.
The $\OXT$ protocol supports conjunctive queries of the form $t_1\wedge t_2\wedge ...\wedge t_n$ natively, but it can be extended to support the boolean query of the form $t_1\wedge \phi(t_2, ..., t_n)$, where $t_1$ is the `s-term', and $\phi(t_2, ..., t_n)$ is an arbitrary boolean expression~\cite{cash2013highly}. As shown in Algorithm~\ref{alg:boolean}, the extended protocol follows the basic steps to obtain search tokens and search in $\TSet$ and $\XSet$ interactively. 
Nevertheless, it introduces additional steps (line 3, 13, 15--17 in Algorithm~\ref{alg:boolean}) to solve the boolean expression $\phi(t_2, ..., t_n)$. Specifically, $\SF$ substitutes all indexing terms $t_i$ to boolean variables $v_i$ ($i=2, ..., n$) and generates a boolean function $\hat{\phi}(v_2, ..., v_n)$. $\SF$ then sends $\hat{\phi}(v_2, ..., v_n)$ to $\IS$.
 $\IS$ sets the value of $v_i$ to the truth values of $h(\id_c, t_i)\in\XSet$. Then, it evaluates $\hat{\phi}(v_2, ..., v_n)$ and returns $\mathbf{E}_c$ as a result if $\hat{\phi}$ outputs true.

The $\mathrm{Boolean\ Query}$ algorithm can be utilised to enable set operations of social search queries as shown in Table~\ref{tlb:atomic}, which will be discussed in the following section.

\noindent{\bf Security.} In cryptographic terms, the $\OXT$ protocol is proved to be $\Leakage_{\OXT}$-semantically-secure against adaptive attacks, where $\Leakage_{\OXT}$ is the leakage function defined in~\cite{cash2013highly}. It ensures that the untrusted server only learns the information defined in the leakage function, but no other information about the query and underlying dataset. We refer the reader to Section~\ref{subsec:oxt} for more details.


\begin{algorithm}[!t]
\caption{Boolean Query}
\label{alg:boolean}
\begin{algorithmic}[1]
\Require $\TSet, \XSet$, Query $(t_1\wedge \phi(t_2, ..., t_n))$ with s-term $t_1$
\Ensure Encrypted Result $R$
\Function{BooleanQuery}{$\TSet, \XSet, \bar{t}, \phi$}($\bar{t}$ is the indexing term list $(t_1, ..., t_n)$, and $\phi$ is an arbitrary boolean expression)
	\State $\SF$ inputs indexing term $\bar{t}, \phi$, and $\IS$ inputs $\TSet, \XSet$;
	\State $\SF$ initialise a boolean expression $\hat{\phi}(v_2, ..., v_n)$ from $\phi$ and sends it to $\IS$;
	\State $\SF, \IS$ runs $\TSet(t_1)\leftarrow \textsc{IndexAccess}(\TSet, t_1)$;
	\State $\IS$ parses $\TSet(t_1)$ to $\{\mathbf{E}\}$;
	\For{$l=2:n$}
	\State $\SF$ computes $\xtoken(t_1,t_l)$;
	\State $\SF$ sends $\xtoken(t_1,t_l)$ to $\IS$;
	\EndFor
	\State $\IS$ initialises $R\leftarrow\{\}$;
	\For{$c=1:|\{\mathbf{E}\}|$}
		\For{$l=2:n$}
			\State $\IS$ uses $\xtoken(t_1,t_l)$ to compute $h(\id_c, t_l)$;
			\State $\IS$ lets $v_l=(h(\id_c, t_l)\in\XSet)$;
		\EndFor
		\If{$\hat{\phi}(v_2, ..., v_n)=`True'$}
			\State $\IS$ adds $\mathbf{E}_c$ in $R$;
		\EndIf
	\EndFor
	\State\Return $R$;
\EndFunction
\end{algorithmic}
\end{algorithm}
\subsubsection{Arithmetic} 
\system~uses {\it Arithmetic} operations to support complex scoring functions over the retrieved content from {\it Set Operations}.
{\it Arithmetic} operations in \system~involve the secure two-party computation between two $\ISC$s.
%
Here, we introduce the simplest model of \system, where each $\ISC$ only has one $IS$, for ease of presentation on how to use additive shares (see Section~\ref{subsec:secureComp} for detailed definition) to compute addition and multiplication under two-party setting. 
Note that this model can be extended to support multiple pairs of $\IS$s.

In \system, the posting list is generalised as a matrix and the arithmetic operations are evaluated over the matrix.
The reason for that is, instead of running the scoring function with arithmetic operations multiple times for each item of the posting list, the batch processing can reduce the system overhead and support scoring algorithms in parallel.
We denote the matrix of sort-keys returned from a structured information query by $\mathbf{S}$, and the corresponding shared matrix is denoted by $\langle\mathbf{S}\rangle^A$. 
Given two shared matrices $\langle\mathbf{A}\rangle^A$ and $\langle\mathbf{B}\rangle^A$, the addition operation ($Add^A(\langle\mathbf{A}\rangle^A,\langle\mathbf{B}\rangle^A)$) can be evaluated non-interactively by computing $\langle\mathbf{C}\rangle^A=\langle\mathbf{A}\rangle^A+\langle\mathbf{B}\rangle^A$ in each party.
To multiply two shared matrices ($Mul^A(\langle\mathbf{A}\rangle^A,\langle\mathbf{B}\rangle^A)$), two $\IS$s generate the multiplication triplets, which are shared matrices: $\langle\mathbf{X}\rangle^A, \langle\mathbf{Y}\rangle^A, \langle\mathbf{Z}\rangle^A$. $\mathbf{X}$ has the same dimension as $\mathbf{A}$, $\mathbf{Y}$ has the same dimension as $\mathbf{B}$, and $\mathbf{Z}=\mathbf{X}\times\mathbf{Y}\mod 2^l$. 
$\IS_i$ computes $\langle\mathbf{E}\rangle_i^A=\langle\mathbf{A}\rangle_i^A-\langle\mathbf{X}\rangle_i^A$ and $\langle\mathbf{F}\rangle_i^A=\langle\mathbf{B}\rangle_i^A-\langle\mathbf{Y}\rangle_i^A$, and sends it to its counter-party. 
Both parties then recover $\mathbf{E}, \mathbf{F}$ and let $\langle\mathbf{C}\rangle_i^A=i\cdot\mathbf{E}\times\mathbf{F}+\langle\mathbf{X}\rangle_i^A\times\mathbf{F}+\mathbf{E}\times\langle\mathbf{Y}\rangle_i^A+\langle\mathbf{Z}\rangle_i^A$.

The multiplication operation relies on the triplets, which should be generated before the actual computation. 
In addition, each party keeps their $\langle\mathbf{X}\rangle^A, \langle\mathbf{Y}\rangle^A$ in secret during the generation process, otherwise, they can recover $\mathbf{A}, \mathbf{B}$ after two parties exchanged $\langle\mathbf{E}\rangle^A, \langle\mathbf{F}\rangle^A$. 
Thus, \system~introduces a secure offline protocol~\cite{mohassel2017secureml} to generate the triplets via OT, it utilises the following relationship: $\mathbf{Z}=\langle\mathbf{X}\rangle_0^A\times\langle\mathbf{Y}\rangle_0^A+\langle\mathbf{X}\rangle_0^A\times\langle\mathbf{Y}\rangle_1^A+\langle\mathbf{X}\rangle_1^A\times\langle\mathbf{Y}\rangle_0^A+\langle\mathbf{X}\rangle_1^A\times\langle\mathbf{Y}\rangle_1^A$ to compute the shares of $\mathbf{Z}$. 
The resulting offline protocol is only required to compute the shares of $\langle\mathbf{X}\rangle_0^A\times\langle\mathbf{Y}\rangle_1^A$ and $\langle\mathbf{X}\rangle_1^A\times\langle\mathbf{Y}\rangle_0^A$ as the other two terms can be computed locally.

We illustrate the computing process of $\langle\mathbf{X}\rangle_0^A\times\langle\mathbf{Y}\rangle_1^A$ in the offline protocol.
The basic step of the offline protocol is to use $\langle\mathbf{X}\rangle^A_0$ and a column from $\langle\mathbf{Y}\rangle^A_1$ to compute the share of their product. 
This is repeated for each column in $\langle\mathbf{Y}\rangle^A$ to generate $\langle\mathbf{X}\rangle_0^A\times\langle\mathbf{Y}\rangle_1^A$.
Therefore, for simplicity, we focus on the above basic step:
We assume that the size of $\langle\mathbf{X}\rangle^A$ is $\mathrm{s}*\mathrm{t}$ and we denote each element in $\langle\mathbf{X}\rangle^A$ as $\langle x_{i,j}\rangle^A_0$, $i=1, ..., \mathrm{s}$ and $j=1, ..., \mathrm{t}$. 
In addition, we assume each column of $\langle\mathbf{Y}\rangle^A$ has $\mathrm{t}$ elements, which are denoted as $\langle y_j\rangle^A_1$, $j=1, ..., \mathrm{t}$.
The computation process is listed as follows:
\begin{itemize}[leftmargin=*]
\setlength{\itemsep}{0pt}
\setlength{\parsep}{0pt}
\setlength{\parskip}{0pt}

\item $\IS_0$ runs a correlated-OT protocol (COT)~\cite{asharov2013more}, and sets the correlation function to $f_b(x)=\langle x_{i,j}\rangle^A_0\cdot 2^b+x \mod 2^l$ for $b=1, ..., l$.
\item For each bit $b$ of $\langle y_j\rangle^A_1$, $\IS_0$ chooses a random value $r_b$ for each bit and runs $COT(r_b, f_b(r_b);\langle y_j\rangle^A_1[b])$ with $\IS_1$.
\item If $\langle y_j\rangle^A_1[b]=0$, $\IS_1$ gets $r_b \mod 2^l$; If $\langle y_j\rangle^A_1[b]=1$, $\IS_1$ gets $ f_b(r_b)=\langle x_{i,j}\rangle^A_0\cdot 2^b+r_b \mod 2^l$. It is equivalent to get $\langle y_j\rangle^A_1[b]\cdot \langle x_{i,j}\rangle^A_0\cdot 2^b+r_b \mod 2^l$ in $\IS_1$ side.
\item $\IS_1$ sets $\langle \langle x_{i,j}\rangle^A_0\cdot \langle y_j\rangle^A_1\rangle_1^A=\Sigma_{b=1}^l \langle y_j\rangle^A_1[b]\cdot \langle x_{i,j}\rangle^A_0\cdot 2^b+r_b \mod 2^l$, and $\IS_0$ sets $\langle \langle x_{i,j}\rangle^A_0\cdot \langle y_j\rangle^A_1\rangle_0^A=\Sigma_{b=1}^l (-r_b) \mod 2^l$. 
\end{itemize}
%
%
%
%
%
%
After computing $\langle \langle x_{i,j}\rangle^A_0\cdot \langle y_j\rangle^A_1\rangle^A$, the $j$-th element of the $i$-th row in $\langle\langle\mathbf{X}\rangle_0^A\times\langle\mathbf{Y}\rangle_1^A\rangle^A$ is $\Sigma_{j=1}^\mathrm{t} \langle \langle x_{i,j}\rangle^A_0\cdot \langle y_j\rangle^A_1\rangle^A$. 
Analogously, $\IS_0$ and $\IS_1$ can compute the share of $\langle\mathbf{X}\rangle_1^A\times\langle\mathbf{Y}\rangle_0^A$ in the same way.
%
%

\noindent{\bf Security.} Additive sharing scheme offers security guarantees to {\it Arithmetic} operations in \system~via its computational indistinguishable property.
More specific, as discussed in~\cite{pullonen2012design}, the scheme can create a uniformly distributed input and output to protects the original input/output of {\it Arithmetic} operation under the threat model of \system, i.e., semi-honest but non-colluding two-party.

\begin{figure}[!t]
\includegraphics[width=\linewidth]{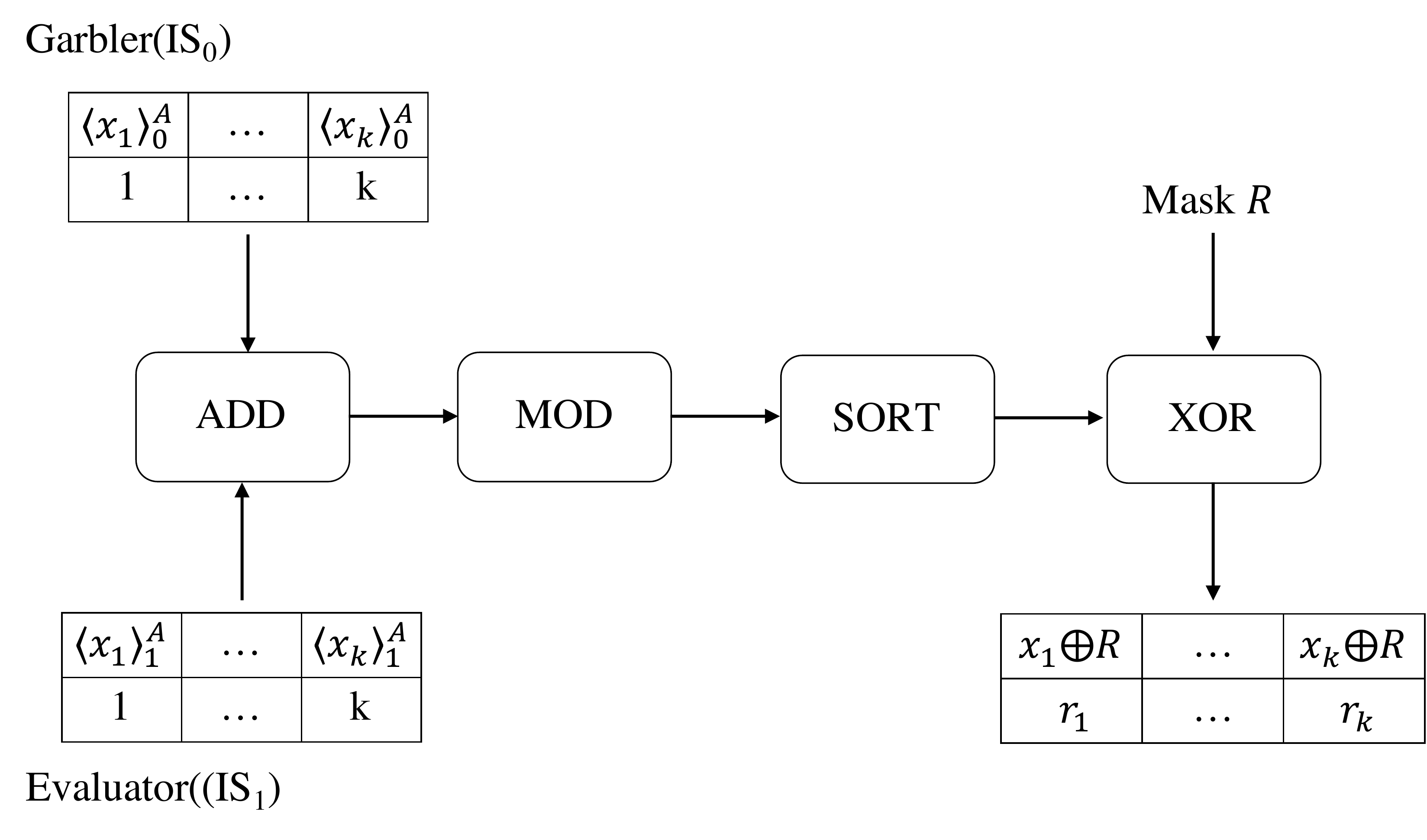}
\caption{Local sorting process for one pair of garbler and evaluator. The input of the garbler is at the top, and the input/output of the evaluator is on the bottom.}
\vspace{-10pt}
\label{fig:localsorting}
\end{figure}
\subsubsection{Sorting} This is a required operation in order to rank the computed scores from {\it Arithmetic} operations.
A naive solution is to recover all scores from additive shares in $\SF$ and sort them as plaintext. However, transmitting all rank results to $\SF$ is a bandwidth-consuming task, the sort operation can be very inefficient as the result. 
Therefore, \system~chooses to mix the additive sharing scheme and Yao's Garbled Circuit (see Section~\ref{subsec:secureComp} for details) to support arithmetic operations and comparison at the same time, as it avoids the communication overhead from sending the shares back to $\SF$.
To protect the privacy of score values, the generated circuit should have a fixed sequence of comparison for a given size of inputs (i.e., achieving the trace-oblivious), and it should not reveal the actual scoring value after circuit evaluation.

\noindent{\bf Local Sorting.} To enable sorting on $\IS$s, \system~leverages an efficient scheme in~\cite{demmler2015aby} to switch from additive sharing to Yao's sharing. It then adopts the sorting network~\cite{batcher1968sorting} to generate the optimised sorting circuit. Finally, the garbler concatenates the sorting network with an XOR gate and applies a random mask $R$ to mask the score values. As a result, the evaluator can use decode table $dec$ to figure out the rank, but it does not know the score values.
Thus, the local sorting algorithm in \system~can be divided into five phases. Figure~\ref{fig:localsorting} illustrates the process of local sorting.

Given $\IS_0$ as the garbler and $\IS_1$ as the evaluator, both parties pre-share a scoring vector $\mathsf{x}=\{\langle x_i\rangle^A\}_{i=1}^\mathrm{k}$, \system~runs the protocol $\textsc{LocalSort}(\mathsf{x})$ to sort the vector and returns a sorted vector in descending order of $\mathsf{x}$, the protocol can be summarised as follows:
\begin{itemize}[leftmargin=*]
\setlength{\itemsep}{0pt}
\setlength{\parsep}{0pt}
\setlength{\parskip}{0pt}
\item Phase 1: $\IS_0$ runs $\GC$ to generate the circuit in Figure~\ref{fig:localsorting} as well as its decode table $dec$. It then sends the circuit and the decode table $dec$ to $\IS_1$. Doing so ensures that $\IS_1$ only can see the final result with random mask $R$.
\item Phase 2: $\IS_0$ sends the encoded inputs of its additive shares $\{\langle x_i\rangle_0^A\}_{i=1}^\mathrm{k}$ with a payload vector $\{i\}_{i=1}^\mathrm{k}$ indicating the position. This prevents $\IS_1$ from learning the additive shares of $\IS_0$.
\item Phase 3: $\IS_1$ retrieves the encoded inputs of its additive shares $\{\langle x_i\rangle_1^A\}_{i=1}^\mathrm{k}$ and payload vector from $\IS_0$ via OT protocol. This prevents $\IS_0$ from learning the additive shares of $\IS_1$.
\item Phase 4: $\IS_0$ generates the encoded input of a random mask $R$ to perform the last XOR gate to protect the vector.
\item Phase 5: $\IS_1$ uses the given inputs to evaluate the circuit, and uses $dec$ to decode the outputs. 
\end{itemize}

Since the circuit puts a mask $R$ after sorting, $\IS_1$ only gets the ranking $\{r_i\}_{i=1}^k$ without knowing the actual scores.

\begin{figure}[!t]
\includegraphics[width=\linewidth]{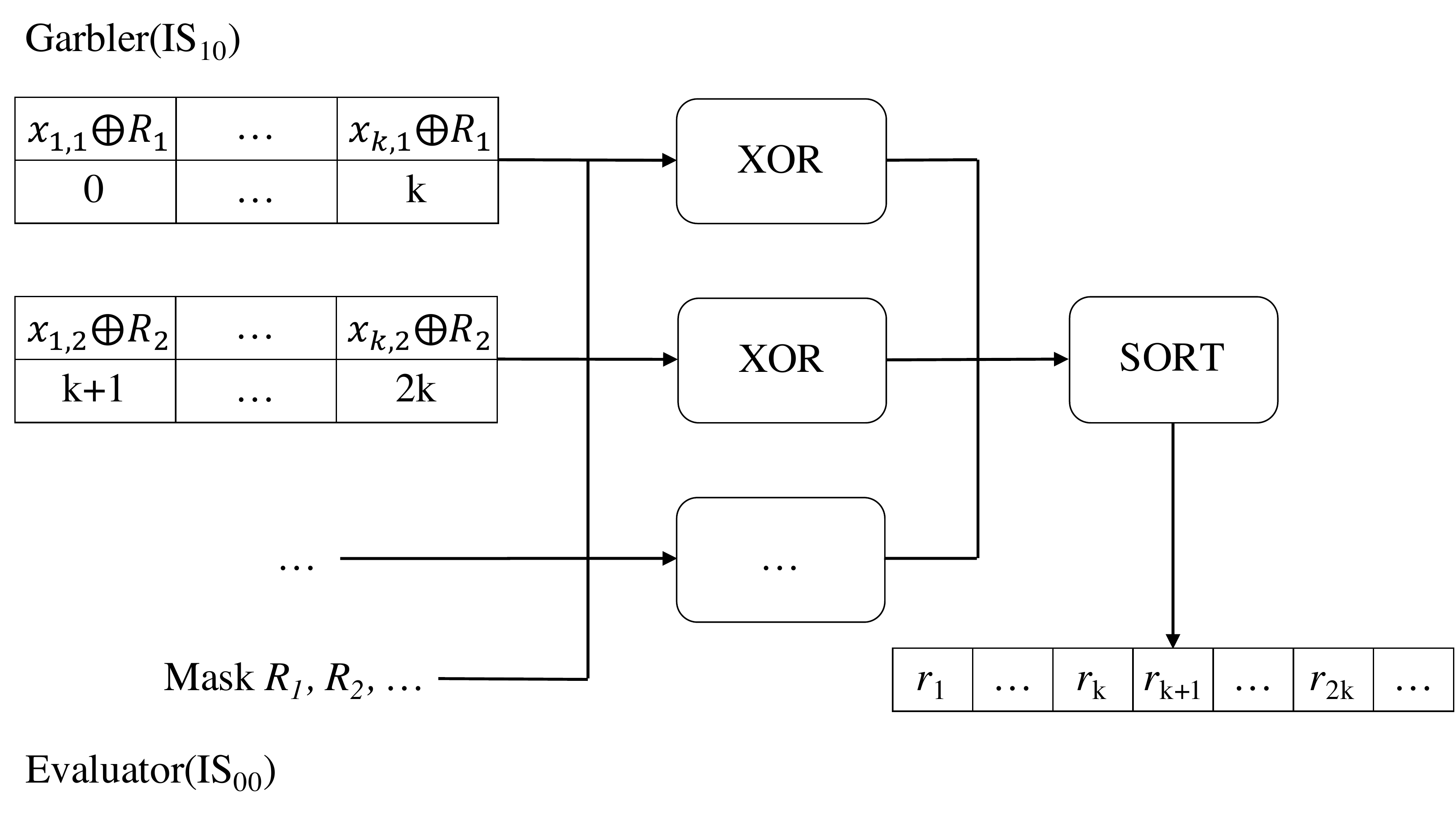}
\caption{Global sorting process in the coordinators. The input of the garbler is the masked score vector with a payload that indicates the position of the score in vector, and the input of the evaluator is the random masks.}
\vspace{-10pt}
\label{fig:globalsorting}
\end{figure}
\noindent{\bf Global Sorting.} The above sorting strategy is a suitable and efficient solution for the simplest model, i.e., only one $\IS$ in each $\ISC$. However, it can be problematic when each $\ISC$ has several $\IS$s. In this case, no $\IS$ can provide a full sorted list as each $\IS$ only has a disjoint part of the whole graph. Hence, $\SF$ still needs to perform another inefficient plaintext sorting. 

Therefore, \system~uses a specific protocol which runs by a chosen coordinator of each $\ISC$. 
The protocol can perform an extra round of sorting upon the results from local sort while keeping the scoring value in secret. 
Assuming that each $\ISC$ has $n$ different $\IS$s, $\IS_{00}$ and $\IS_{10}$ are chosen to be the coordinators for $\ISC_0$ and $\ISC_1$, respectively. 
After local sorting, evaluator in $\ISC_1$ sends a vector of masked scoring values $\{x_{i,j}\oplus R_j\}$ where $i=1, ..., \mathrm{k}$ and $j=1, ..., \mathrm{n}$ to $\IS_{10}$ and garbler in $\ISC_0$ sends the masks $\{R_j\}_{j=1}^\mathrm{n}$ to $\IS_{00}$. 
In the global sorting, \system~switches the roles of $\IS_{10}$ and $\IS_{00}$ (i.e., $\IS_{10}$ is the garbler, and $\IS_{00}$ is the evaluator) for two reasons:
It prevents $\IS_{10}$ from evaluating two circuits at the same time, as $\IS_{10}$ also needs to evaluate another local sorting circuit for its partition.
Furthermore, it facilitates pipeline data processing. The partial result of each $\IS$ can be sent to $\IS_{10}$ and $\IS_{00}$ for global sorting separately. Once the $\IS_{10}$ and $\IS_{00}$ get the first result, they can start to run encoding and OT.
The protocol $\textsc{GlobalSort}(\mathsf{x})$ is summarised as follows: 
\begin{itemize}[leftmargin=*]
\setlength{\itemsep}{0pt}
\setlength{\parsep}{0pt}
\setlength{\parskip}{0pt}
\item Phase 1: $\IS_{10}$ runs $\GC$ to generate the circuit in Figure~\ref{fig:globalsorting} and the decode table $dec$. It then sends the circuit and $dec$ to $\IS_{00}$.
\item Phase 2: $\IS_{10}$ sends the encoded inputs $\{x_{i,j}\oplus R_j\}$ where $i=1, ..., \mathrm{k}$ and $j=1, ..., \mathrm{n}$ with a payload vector $\{i\}_{i=1}^\mathrm{n\cdot k}$.
\item Phase 3: $\IS_{00}$ retrieves encoded masks $\{R_j\}_{j=1}^\mathrm{n}$ via OT protocol.
\item Phase 4: $\IS_{00}$ uses the given inputs to evaluate the circuit, and uses $dec$ to decode the outputs. The result is sent to $\SF$ in descending order.
\end{itemize}

\noindent{\bf Security.} The security properties of Garbled Circuit~\cite{bellare2012foundations} and OT~\cite{asharov2013more} ensure the security of {\it Sorting} in \system~in the following three aspects: 
Firstly, no adversary can learn the input of its counter-party when sorting (i.e., the other additive share in $\textsc{LocalSort}(\mathsf{x})$ and the masked score/mask in $\textsc{GlobalSort}(\mathsf{x})$).
Secondly, the output of $\textsc{LocalSort}(\mathsf{x})$ is masked by a one-time mask, which is a uniformly random number. It protects the original score vector because the evaluator only learns the masked values from output.
Finally, for the output of $\textsc{GlobalSort}(\mathsf{x})$, only the decode table of rank is sent to the evaluator, which also ensures that the evaluator only learns global rank without knowing the actual ranking score.

\section{Query Realisation}\label{sec:queries}
In \system, $\SF$ receives queries as the query strings in the form of s-expression. It is composed of several operators to describe the set of results the client wishes to receive (see Table~\ref{tlb:atomic}).
In the following sections, we introduce the operators in \system~and their security properties.

\subsection{Graph Operators}\label{subsec:operator}
\system~uses atomic operations in Section~\ref{subsec:atomic} to realise all operators in Table~\ref{tlb:atomic}. 
In this section, we present the detailed constructions of these operators. 
Note that we only consider the operators as the outermost operators, i.e., they are not nested in any other query strings, because the query plan generation highly depends on the outermost operators.

\noindent{\bf term.} The \textbf{term} operator runs an \textit{Index Access} operation to retrieve a posting list. 
In addition, if there is a requirement to rank the result, the \textit{Sorting} operation is able to return a sorted posting list which puts the record with higher relevance at the beginning of the list.

\noindent{\bf and.} This operator is natively supported by the \textsc{BooleanQuery} algorithm. 
As mentioned in Section~\ref{subsec:atomic}, conjunctive queries with nested queries (e.g., $t_1\wedge \phi(t_2, ..., t_n)$) are processed by evaluating the boolean expression $\phi$ in $\XSet$. It is obvious that the \textbf{and} operator is executed in a sub-linear time, as its complexity is proportional to the size of $\TSet(t_1)$.

\noindent{\bf difference.} The \textbf{difference} operator is extended from \textsc{BooleanQuery} algorithm. Considering the query (\textbf{difference} \textbf{friend:3} (\textbf{and} \textbf{friend:1} \textbf{friend:2})) from Table~\ref{tlb:atomic}, it aims to find the friends of $\textbf{user}_3$, who are neither $\textbf{user}_1$'s friends nor $\textbf{user}_2$'s friends. The boolean expression, in this case, is $\textbf{friend:1}\land \textbf{friend:2}$, but the results that satisfy the expression are removed from the results of the query $(\textbf{term}\ \textbf{friend:3})$.
In summary, \textbf{difference} operator excludes the results that satisfy the boolean expression $\phi$. Therefore, the s-expression with \textbf{difference} operator is represented as
$t_1\wedge \neg\phi(t_2, ..., t_n)$.
Comparing with the \textbf{and} operator, it returns the results only if the boolean expression $\phi$ returns false instead of true.

\noindent{\bf or.} The complexity of the original approach for processing disjunctive queries is linear to the size of database~\cite{cash2013highly}. 
To achieve a sub-linear time complexity, we leverage the above \textbf{difference} and \textbf{term} operators to build a new disjunctive query operator. 
In particular, the s-expression starts with \textbf{or} operator can be processed via a list of \textbf{difference} expressions and an additional \textbf{term} expression. 
For instance, if a disjunctive query has three indexing terms: $t_1, t_2, t_3$, the corresponding s-expression is ($\textbf{or}\ t_1\ t_2\ t_3$), and it is parsed as:
$(\textbf{difference}\ t_1\ (\textbf{or}\ t_2\ t_3)), (\textbf{difference}\ t_2\ t_3), (\textbf{term}\ t_3)$.
The above three s-expressions return three different sets of results, and the composite of them is the final result of \textbf{or} operator. The correctness of the above approach can be easily proved by the set operation:
$t_1\lor t_2\lor t_3 = (t_1\textbackslash(t_2\lor t_3)) \lor (t_2\textbackslash t_3) \lor t_3$.

In general, a disjunctive s-expression with $n$ indexing terms can be rewritten as $n-1$ s-expressions with \textbf{difference} operator and $1$ s-expressions with \textbf{term} operator. The complexity is proportional to $|t|\cdot M$, where $|t|$ is the number of disjunctive indexing terms and $M$ is the result size of the most frequent term $\max(\{|\TSet(t_i)|\}_{i=1}^n)$.

\subsection{Apply Operator}
The \textbf{apply} is a unique operator in Unicorn~\cite{curtiss2013unicorn}, which enables graph-traversal. 
The basic idea is to retrieve the results of nested queries and use these results to construct and execute a new query. 
For example, given an s-expression (\textbf{apply friend:} \textbf{friend:}$\id_1$), $\SF$ issues a query (\textbf{term}\ \textbf{friend:}$\id_1$) and collects $N$ users, it then generates the second query
$(\textbf{or}\ \textbf{friend:}\id_{1,1}\ ...\ \textbf{friend:}\id_{1,N})$ to get the entities that are more than one edge away from the user $\id_1$ in the encrypted graph-structured data.

\system~defines a query structure to construct \textbf{apply} operator. 
In details, the query structure is a tuple of (\textit{prefix}, \textit{s}, \textit{filter}), where the \textit{prefix} (e.g., \textbf{friend:}) is prepended to the given $\id$ to form the indexing terms, \textit{s} is an s-expression with $N$ indexing terms (e.g., (\textbf{term}\ \textbf{?})), and \textit{filter} indicates the ranking algorithm for its results.
To execute an \textbf{apply} operator, $\SF$ pre-processes the input $\id$ with a given \textit{prefix} in the query structure and uses processed $\id$ to execute the s-expression \textit{s} from the query structure.
$\ISC$~handles the query from the s-expression and applies the designated \textit{filter} in the query structure to refine the result.
Consequently, $\SF$ can retrieve a list of $\id$ as the result of the nested query structure. \system~leverages as input the retrieved $\id$ and outer query structure to repeat the above procedure until it reaches the outermost query structure.
Algorithm~\ref{alg:apply} gives the detailed implementation of the \textbf{apply}.

\begin{algorithm}[!t]
\caption{Apply}
\label{alg:apply}
\begin{algorithmic}[1]
\Require Outer Query Structure $QS_O$, Nested Query Structure $QS_N$, Array of $\id$
\Ensure Encrypted Result $O$
\Function{Apply}{$QS_O, QS_N, \id$}(in $\SF$)
\For{$i=1:QS_N.\textit{s}.size$}
\State$term\leftarrow QS_N.\textit{prefix}||\id[i]$;
\State $QS_N.\textit{s}[i]\leftarrow term$;
\EndFor
\State $o\leftarrow \mathrm{Search}(QS_N.\textit{s}, QS_N.\textit{filter})$;
\For{$i=1:o.size$}
\State$term\leftarrow QS_O.\textit{prefix}||o[i].id$;
\State $QS_O.\textit{s}[i]\leftarrow term$;
\EndFor
\State $O\leftarrow \mathrm{Search}(QS_O.\textit{s}, QS_O.\textit{filter})$;
\State\Return $O$;
\EndFunction
\end{algorithmic}
\begin{algorithmic}[1]
\Require Query $q$, Result Filter $f$
\Ensure Encrypted Result $r$
\Function{Search}{q, f}(in $\IS$)
\State $r\leftarrow f(Execute(q))$;
\State\Return $r$
\EndFunction
\end{algorithmic}
\end{algorithm}

The \textbf{apply} operator processes necessary steps on behalf of its users to improve the efficiency of \system. 
For example, it would be possible for users to ask recommendation from friends: both $\SF$ and client can execute a two-step query to retrieve friend list in advance and issue an additional query to get the recommendation. 
Compared to the latter strategy, the \textbf{apply} operator runs in $\SF$ can highly reduce the workload on the client side: it saves the network latency of transmitting intermediate result between $\SF$ and client, addition to the computational cost of aggregation and regeneration. 
Furthermore, $\SF$ can further optimise the query and adopt the different scoring strategy by giving its semantic context (as shown in the following example).

\noindent{\bf Example: Friend Recommendation.}
The friend recommendation is a good example for the use of the \textbf{apply} operator. 
According to the Homophily theory~\cite{mcpherson2001birds}, people with higher similarity have a higher probability to become friend. 
In this context, the system aims to recommend the friends-of-friends to its user according to the order of similarities. 
Hence, we apply a simple ranking function which returns the sorted similarity value directly for both outer and nested query for this application. 


It is also possible to implement the friends-of-friends query without the \textbf{apply} operator. Intuitively, friends-of-friends also can be treated as an edge type, \system~may explicitly store the friends-of-friends list and use the indexing term $\textbf{friends-of-friends:id}$ to index it. Therefore, the friend recommendation problem is easily processed by the \textbf{term} operator.
However, such a naive solution blows up the memory consumption on $\IS$: as shown in Table~\ref{tlb:FoF}, the estimated size of friends-of-friends posting list is almost 370x larger than the original friend list of a typical user~\cite{curtiss2013unicorn}. 
In \system, each encrypted tuple occupies $56$ bytes memory space (See Section~\ref{subsec:storage} for the detailed discussion), it indicates that the friends-of-friends posting lists for 1 million users consumes $2.4$TB RAM.

The \textbf{apply} operator also reduces the query latency: it is expensive to sort the posting lists of $\textbf{friends-of-friends:id}$ with 48k entities inline in Garbled Circuit (it needs $200$ s to evaluate the corresponding circuit). 
In comparison, introducing an extra round to query enables $\IS$ to truncate the result, and makes the sorting process more efficient. 
For example, if \system~applies a \textit{filter} to return the top $10$ results to $\SF$ for the nested query, the result size can be reduced to around $1000$ users with a higher similarity.
%
%
Also, the reduced result size is moderate to evaluate sort circuit on it. Under our settings in Section~\ref{sec:evaluation}, the system replies either a full result list after $5$-$6$ s or a truncated result list after $1$-$2$ s.

\begin{table}[!t]
	\centering
	\caption{Performance estimation of Friend Recommendation implementations with 1 million users.}
	\label{tlb:FoF}
	\begin{tabular}{lll}
	& {\bf friends-of-friends} & \textbf{apply} \\ 
	\hline\hline
	{\bf Est. \# of users/posting list} & $130$ & $48$k \\
	{\bf Est. Storage overhead} & $7.28$GB & $2.4$TB \\
	{\bf Query delay} & $200$ s & $1\text{-}2$ s\\
	\hline\hline
	\end{tabular}
	\vspace{-10pt}
\end{table}

\subsection{Security Analysis}
In Section~\ref{subsec:atomic}, we discuss the security of each atomic operation. 
%
%
Here, we analyse the security of the overall system. 
Specifically, we formulate the security of \system~based on the prior work of SSE~\cite{curtmola2011searchable,cash2013highly} and further combine the security of additive sharing and Garbled Circuit to depict the security of query operators.  
Throughout the analysis, we consider a query in \system~containing a boolean formula $\phi$ and a tuple of indexing terms $(t_1, t_2, ..., t_n)$. 

\noindent{\bf Overview.} The main idea of analysing the security of \system~is similar to that in SSE scheme~\cite{curtmola2011searchable,cash2013highly}. Specifically, the analysis constructs a simulator of \system~to show that the adversary in \system~only learns the controlled leakage parameterised by a leakage function $\Leakage$, after querying a vector of queries $\query$.
Note that the simulator of \system~is slightly different from the original SSE simulator, we outline these different points as the sketch of our security analysis.
Firstly, we update the leakage function of SSE ($\OXT$ in \system) to additionally capture the ranks leaked in query results.
Secondly, we slightly modify the capability of adversaries to fit our two-party model: an adversary in our system is able to see the view on the corrupted cluster as well as the output of the counter-party. 
Under the new adversary model, the joint distribution of the outputs of both the adversary and the counter-party can be properly simulated by an efficient algorithm with the updated leakage function.
Finally, as \system~has the submodules implemented by SSE, additive share scheme, and Yao's Garbled Circuit, the simulator of \system~can be constructed by combining the simulators of these submodules.
For instance, our simulator uses the output of SSE simulator as the input of garbled circuit simulator in order to simulate the query operators with structured data access and sorting.

Due to the page limit, the updated leakage function is given in Appendix~\ref{subsec:leak} and the detailed proof is in Appendix~\ref{subsec:proof}.

\noindent{\bf Discussion.} Note that there exist some emerging threats against the building blocks of \system. Regarding SSE, leakage-abuse attacks~\cite{cash2015leakage,zhang2016all} can help an attacker to explore the information learned during queries. To mitigate them, recent studies on padding countermeasures~\cite{cash2015leakage,bostthwarting} and forward/backward privacy~\cite{bost2017forward,sun2018practical} are proposed and shown to be effective. We leave the integration of these advanced security features to our system as future work. 
Regarding sorting, \system~reveals the rank of the query result. Recent work~\cite{kellaris2016generic, kornaropoulosdata} demonstrates that the underling data values are likely to be reconstructed if an adversary knows ranks and some auxiliary information of queries and datasets. Currently, we do not consider such a strong adversary, and how to fully address the above threat remains as an interesting problem.


%


\section{Implementation}\label{sec:impl}
We implement a prototype system for evaluating the performance of \system. 
To build this prototype, we first realise the cryptographic primitives in Section~\ref{sec:pre}. 
Specifically, we use the symmetric primitives, i.e., AES-CMAC and AES-CBC, from Bouncy Castle Crypto APIs~\cite{castle2007legion}. 
In addition, we use a built-in curve from Java Pairing-based Cryptography (JPBC)~\cite{CaroI11} library (Type A curve) to support the group operations in $\OXT$. 
The security parameter of symmetric key encryption schemes is 128-bit, and the security parameter of the elliptical curve cryptographic scheme is 160-bit.
Regarding the secure two-party computation, we set the field size to $2^{31}$. 
Therefore, we can use regular arithmetics on Java integer type to implement the modulo operations, as it is significantly faster than the native modulo operation in Java BigInteger type (i.e., we observed that it is 50x faster). 
We involve this optimisation into the implementation of additive sharing scheme in the finite field $\mathbb{Z}_{2^{31}}$, the addition (multiplication) operations is calculated by several regular addition and multiplication operations with the modulo operation.
Oblivious Transfer and Garbled Circuit are implemented by using FlexSC~\cite{xiao2018gc}. It implements the extended OTs in~\cite{asharov2013more} and several optimisations for the garbled circuit, which make it a practical primitive under Java environment.

The prototype system consists of three main components: the {\it encrypted database generator}, the {\it query planner} in $\SF$ and the {\it index server daemon} in $\ISC$.
The {\it encrypted database generator} is running on a cluster with Hadoop~\cite{hadoop}.
It partitions the plaintext data, runs the adapted $\OXT$ to convert the data into encrypted tuples with the additive share of sort-keys and stores these tuples on each $\IS$. 
We leverage Spark~\cite{zaharia2010spark} to execute these tasks in-memory and enable the pipelining data processing to further accelerate this process. 
The generated tuples are stored in the in-memory key-value store Redis~\cite{redis} in the form of $\TSet$ on each $\IS$ for querying purpose later. In addition, the generated $\XSet$ is kept in the external storage of each $\IS$ to support the set operations.
All queries are handled by the {\it query planner} and {\it index server daemons} by following the query processing flow in Section~\ref{subsec:arch}. Thrift thread pool proxy~\cite{slee2007thrift} is deployed to handle the queries in {\it index server daemons}.

To improve the runtime performance of our prototype, each posting list is segmented into fixed-size blocks indexed by its $\stag(t)$ and a block counter $c$ for the $\stag$. 
As the final result of the block counter indicates the total number of blocks for each $\stag$, it is also stored in Redis after the whole posting list is converted to encrypted tuples. Those counters enable $\IS$ to retrieve multiple tuples in parallel.
We also introduce a startup process for $\OXT$ protocol and secure two-party computation in {\it index server daemons}.
In terms of the $\OXT$ matching, the index server daemon creates a Bloom Filter~\cite{Bloom70} to load the $\XSet$ into memory during the startup process. In our prototype, we deploy the Bloom filter from Alexandr Nikitin as it is the fastest Bloom filter implementation for JVM~\cite{alexandrnikitin2017bloomfilter}. We set the false positive rate to $10^{-6}$, and the generated Bloom Filter only occupies a small fraction of $\IS$ memory.
Besides generating the Bloom Filter, each index server also pre-computes several multiplication triplets and sorting circuits and periodically refreshes it to avoid extra computational cost on-the-fly.

Our prototype system implementation consists of four main modules with roughly 3000 lines of Java code, we also implement a test module with another 1000 lines of Java code.

\begin{table}[!t]
\centering
	\caption{Statistics of Youtube social network dataset.}
	\label{tlb:statistic}
	\begin{tabular}{|c|c|c|c|}
	\hline
	Node type& \# of nodes & Edge type & \# of edges \\
	\hline
	User & $1157827$ & {\bf friend} & $4945382$\\
	\hline
	Group & $30087$ & {\bf follow} & $293360$ \\
	\hline
	\end{tabular}
	\vspace{-10pt}
\end{table}
\section{Experimental Evaluations}\label{sec:evaluation}
\subsection{Setup}
{\bf Platform.} We deploy the {\it index server daemons} in a cluster ($\ISC$) with 6 virtual machine instances in the Microsoft Azure platform. All VMs are E4-2s v3 instances, configured with 2 Intel Xeon E5-2673 v4 cores, 32GB RAM, 64GB SSD external storage and 40Gbps virtualised NIC. 
Another D16s v3 instance is created in Azure to run $\SF$ with the {\it query planner} and client; it is equipped with 16 Intel Xeon E5-2673 v3/v4 cores, 64GB RAM, 128GB SSD external storage and 40Gbps virtualised NIC. 
We also have the other three E4-2s v3 instances controlled by $\SF$; we use them to run the {\it encrypted database generator} for generating the encrypted index.
All VMs are installed with Ubuntu Server 16.04LTS.

\noindent{\bf Dataset.} We use a Youtube dataset from~\cite{mislove-2007-socialnetworks}, which is an anonymised Youtube user-to-user links and user group memberships network dataset. The detailed statistical summary is given in Table~\ref{tlb:statistic}. 
We recognise the user-to-user links as {\bf friend} edge and user group memberships as {\bf follow} edge from this dataset. The generated posting lists are indexed by above two edge types and user $\id$s.
As the social network in our Youtube dataset is an unweighted network, we randomly generate a weight between $1$ and $100$ for each edge to evaluate the arithmetic and sort operations of \system.

\noindent{\bf Baseline.} To evaluate the performance of \system, we create a graph search system by removing/replacing cryptographic operations in this baseline system. 
Specifically, we leverage hash function to generate $\xtag$ instead of using expensive group operations. 
The index and sort-key are stored in plaintext, which means that the $\IS$ can compute and sort without any network communication for OT and multiplication triplets. 
Finally, the query planner provides the indexing term in plaintext instead of $\stag$ to the $\IS$ as query token. 
Nonetheless, we still use the PRF value of indexing term and the block counter as tuple index, because we want to keep the table structure of $\TSet$ unaltered to make our system comparable to the baseline.
We use this baseline system to evaluate the overhead from cryptographic operations as \system~implements the same operators as Facebook Unicorn~\cite{curtiss2013unicorn}.

\subsection{Evaluation Results}
\begin{table*}[!t]
\centering
	\caption{Benchmark of sorting circuit size and evaluation time, the garbled sorting algorithm is bitonic merging/sorting, we use it to sort $2^l$ vector.}
	\label{tlb:gc}
	\begin{tabular}{cccccccc}
	\hline
	\hline
	Vector length & 2 & 4 & 8 & 16 & 32 & 64 & 128 \\
	\# of AND Gates & 4382 & 9148 & 19448 & 41968 & 91616 & 201664 & 446336\\
	GC evaluation time (ms) & 17.3 & 20.3 & 31.5 & 48.2 & 101.0 & 206.5 & 440.0\\
	GC comm. overhead (MB) & 0.12 & 0.41 & 0.46 & 0.97 & 2.10 & 4.49 & 9.80\\
	\hline
	\hline
	\end{tabular}
	\vspace{-10pt}
\end{table*}

\label{subsec:storage}
\begin{figure}[!t]
\includegraphics[width=\linewidth]{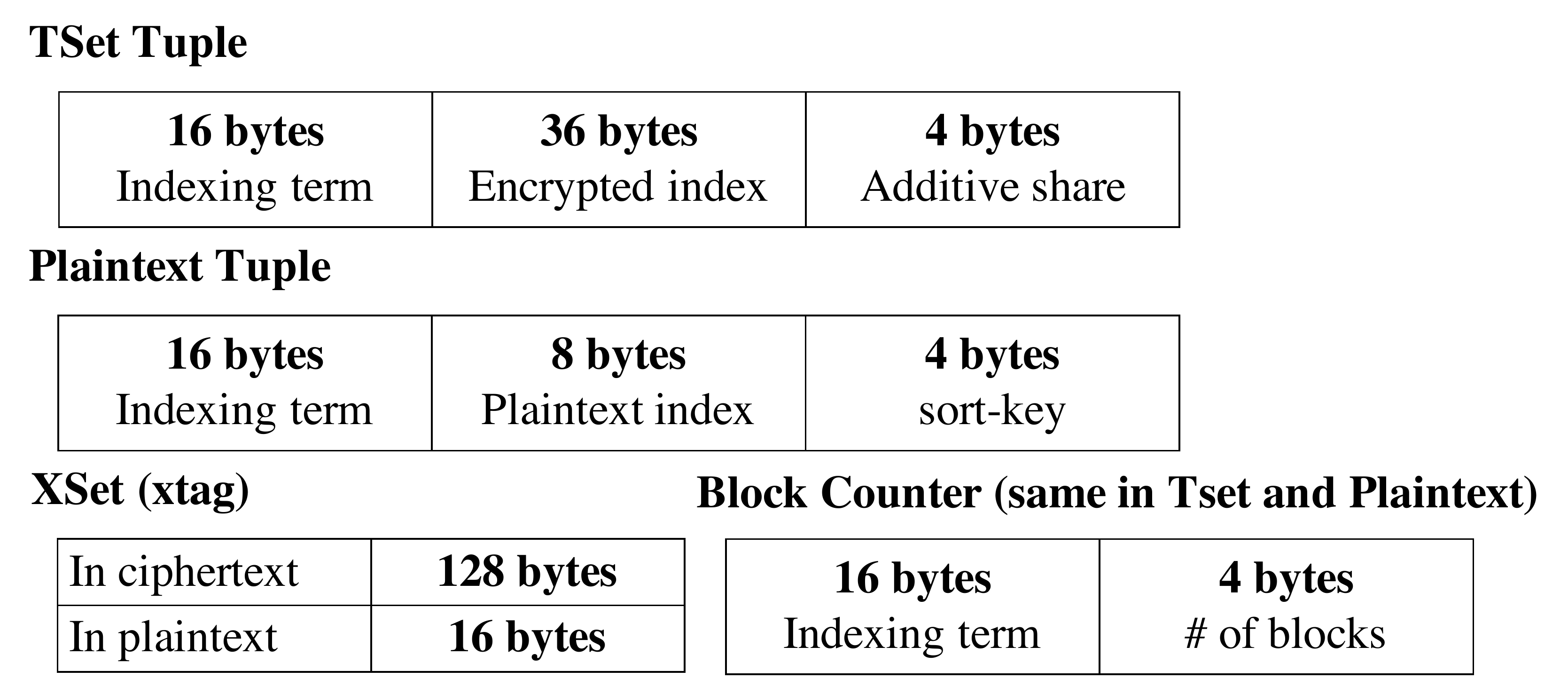}
\caption{A tuple-wise storage overhead comparison between the encrypted database and plaintext database.}
\vspace{-10pt}
\label{fig:storage}
\end{figure}

{\bf EDB Generation.} Firstly, we demonstrate the runtime performance of the {\it encrypted database generator}. 
The generator needs to partition and create additive shares from the original plaintext data, and to generate the adapted $\OXT$ index for each $\IS$.
\system~uses $\SF$ to locally generate the partitions and additive shares for our dataset and then uses the dedicated cluster to generate the encrypted graph index in parallel.
The result on our $5$ million records dataset shows that it only takes $54$ s to pre-process data on $\SF$ and $7.4$ mins to generate the encrypted index via Spark.

\noindent{\bf Storage.} Recall that \system~uses adapted $\OXT$ index to support boolean queries over the encrypted graph, which needs to generate two dedicated data structure (i.e., $\TSet$ and $\XSet$).
As a result, \system~consumes more storage capacity than the baseline system (see Figure~\ref{fig:storage}), because it is required to keep more information (i.e., $\xtag$ in ciphertext), and because it stores encrypted index which is larger than the corresponding plaintext. 
By using the $\TSet$, we observe that our system increases the memory consumption of Redis by $85\%$ (557MB in $\TSet$ and 300MB in plaintext), which is slightly smaller than the theoretical memory consumption overhead ($100\%$ according to Figure~\ref{fig:storage}). 
The reason is that \system~also keeps the number of blocks of each posting list (see Section~\ref{sec:impl}) to accelerate the tuple retrieving process\footnote{If the size of posting list is unknown, the system needs to sequentially retrieve the tuple from blocks, as the key is derived from $\stag$ and block counter. Otherwise, the tuples can be retrieved in parallel.}.
As shown in Figure~\ref{fig:storage}, the block counter requires an additional $20$ bytes of memory consumption for each indexing term both in $\TSet$ and in plaintext. 
It introduces the same extra cost on \system~as well as in the baseline system, and makes the memory consumption overhead smaller than the theoretical expectation.

For $\XSet$ storage overhead, \system~increases it by 17x (1.5GB versus 90MB), mostly due to the fact that the size group element is much larger than a PRF value. 
But the Bloom Filter successfully saves the memory consumption in runtime, because the size of Bloom Filter only depends on the false positive rate and the number of total elements inside (the number of edges in our system)~\cite{Bloom70}, and it is much smaller than $\XSet$ itself. By fixing the false positive rate to $10^{-6}$, the runtime overhead of $\XSet$ in our system is identical to the baseline system (only 18MB in RAM).

\begin{figure}[!t]
\includegraphics[width=\linewidth]{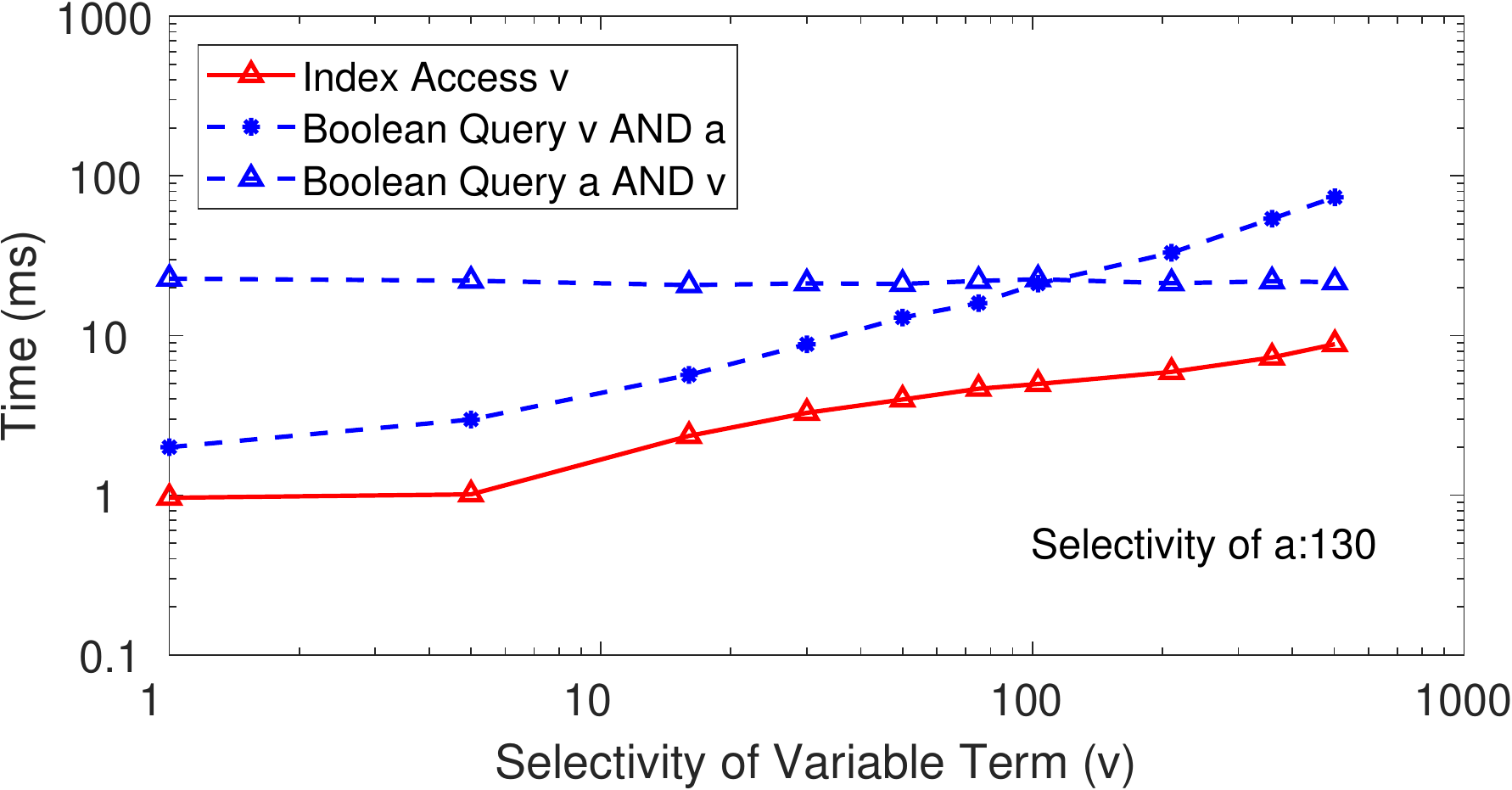}
\caption{Query delay for two-keyword set queries.}
\vspace{-10pt}
\label{fig:twokeywords}
\end{figure}
\noindent{\bf Query Delay.} To understand the query delay introduced by cryptographic primitives, we measure the cryptographic overhead from these cryptographic primitives independently. 
To evaluate the query delay introduced by the set operations, we choose an indexing term {\bf a} from {\bf friend} edges with fixed selectivity ($130$, as the average user has $130$ friends according to Unicorn paper~\cite{curtiss2013unicorn}), and we further choose several variable indexing terms {\bf v} from {\bf friend} edges with selectivity from $1$ to $502$. Figure~\ref{fig:twokeywords} illustrates query delay on Index Access {\bf v} and two variants of two-term Boolean Query. 
In Index Access {\bf v} query, the query only consists of s-term {\bf v}, and the figure shows that its execution time is linear to the size of corresponding posting lists. The other two-term queries combine the previous queries with the fixed term {\bf a}. In the first of these two queries, we use {\bf a} as x-term, each tuple from $\TSet({\bf v})$ should be checked wtih the cost of an exponentiation, which requires $2$ to $79$ ms to process. 
In the last one, {\bf a} is used as s-term, where we observe that the execution time is kept invariable ($20$ ms), irrespective of the variable selectivity of the xterm {\bf v}. 
It demonstrates that \system~can respond a query related to the moderate users with a tiny latency, and it is also able to reply query about popular users with a slightly bigger but still modest delay.

The secure addition and multiplication is supported by additive sharing scheme with a relatively small overhead, because in most cases, the two servers non-interactively do the computation tasks by using regular arithmetic operations, and because the expensive tasks, such as multiplication triplets generation, are pre-computed in the startup process. Figure~\ref{subfig:addition} and \ref{subfig:multiplication} demonstrate the execution delay of addition and multiplication over the different size of vectors. We can see that the addition operation with two vectors containing $10^4$ entities can be done within $2$ ms. For the multiplication operation, it needs $80$ ms to compute the product of $10^4$ entities because it requires several efficient but non-negligible communications with the counter-party.

As the sorting algorithm is implemented by Garbled Circuit, we provide a benchmark about the size of the circuit and the corresponding evaluation time. The results are listed in Table~\ref{tlb:gc}. Note that we do not report the circuit generation time as it is generated in the startup process. The reported result demonstrates the practicality of our sorting strategy: the local sorting generally involves fewer entities after partition ($30$ if the system sorts a result list with $130$ users), which takes $100$ ms to sort. Additionally, the local sort can help to truncate the result before sending it for global sorting, which makes the sorting algorithm in Garbled Circuit more efficient (less than $440$ ms for a vector with $128$ entities).
\begin{figure}[!t]
\includegraphics[width=\linewidth]{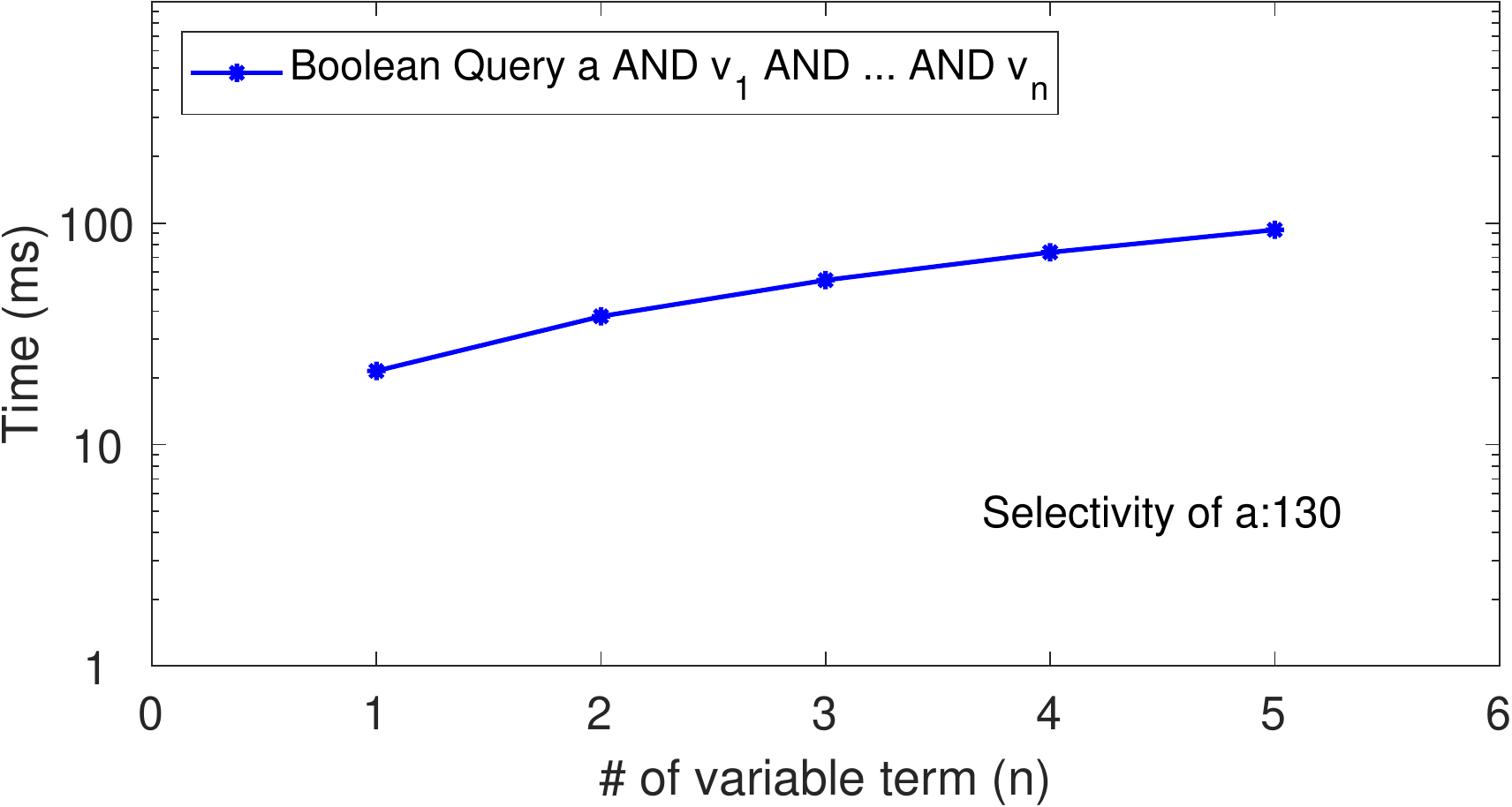}
\caption{Query delay for multiple-keyword set queries.}
\vspace{-10pt}
\label{fig:multikeywords}
\end{figure}
We further examine the query delay from set operations under multiple-keyword setting by using the same {\bf a} as s-term, but we add more variable terms $\{\mathbf{v_n}\}, \mathbf{n}\in[1,5]$ as x-term. Figure~\ref{fig:multikeywords} shows that each additional x-term increases the query delay by $20$ ms, which means a query with $5$ x-terms can still be processed within $100$ ms.

\begin{figure}[!t]
\centering
\subfloat[Addition]{
\label{subfig:addition}
\includegraphics[width=0.45\linewidth]{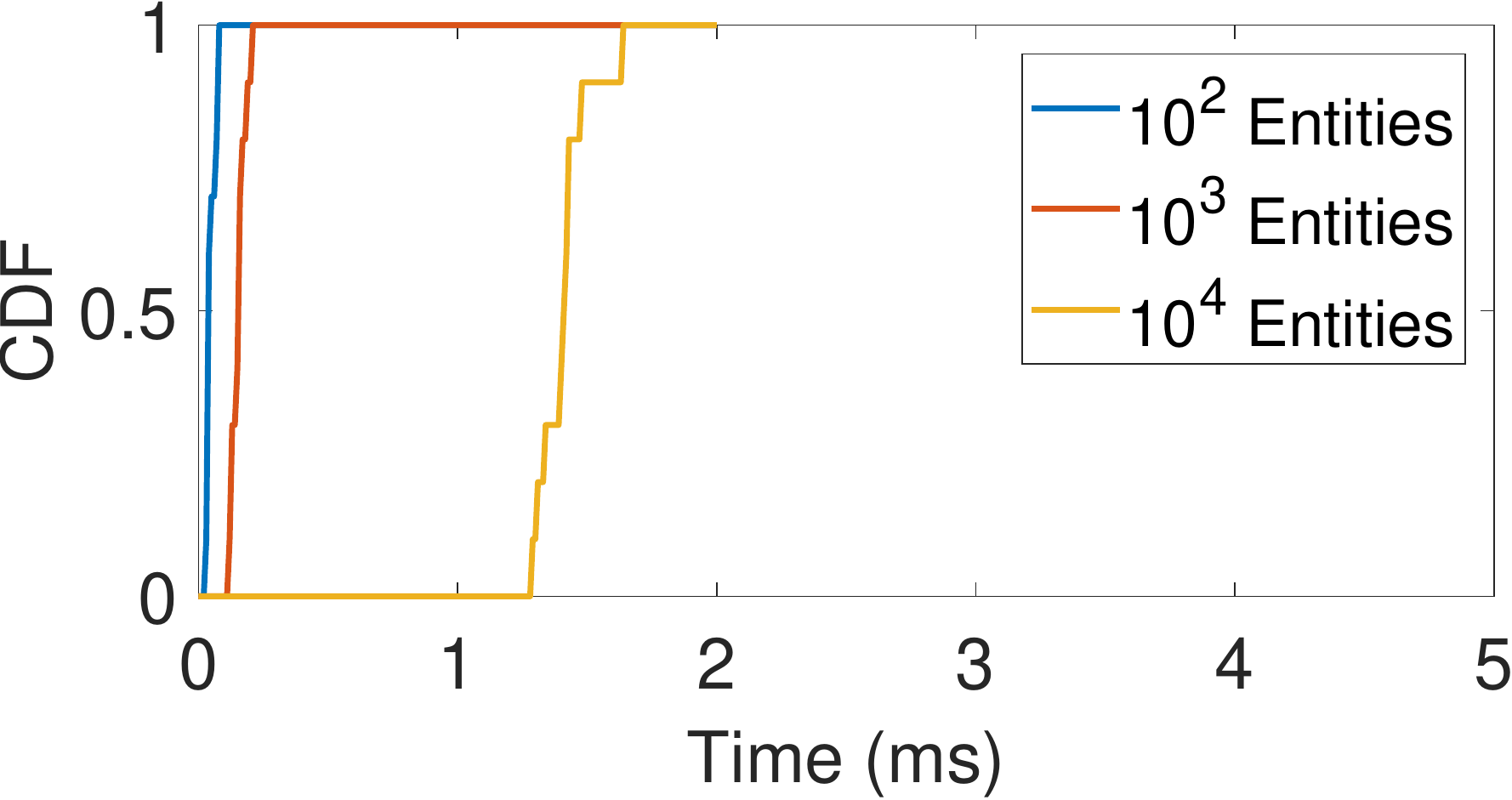}
}
\subfloat[Multiplication]{
\label{subfig:multiplication}
\includegraphics[width=0.45\linewidth]{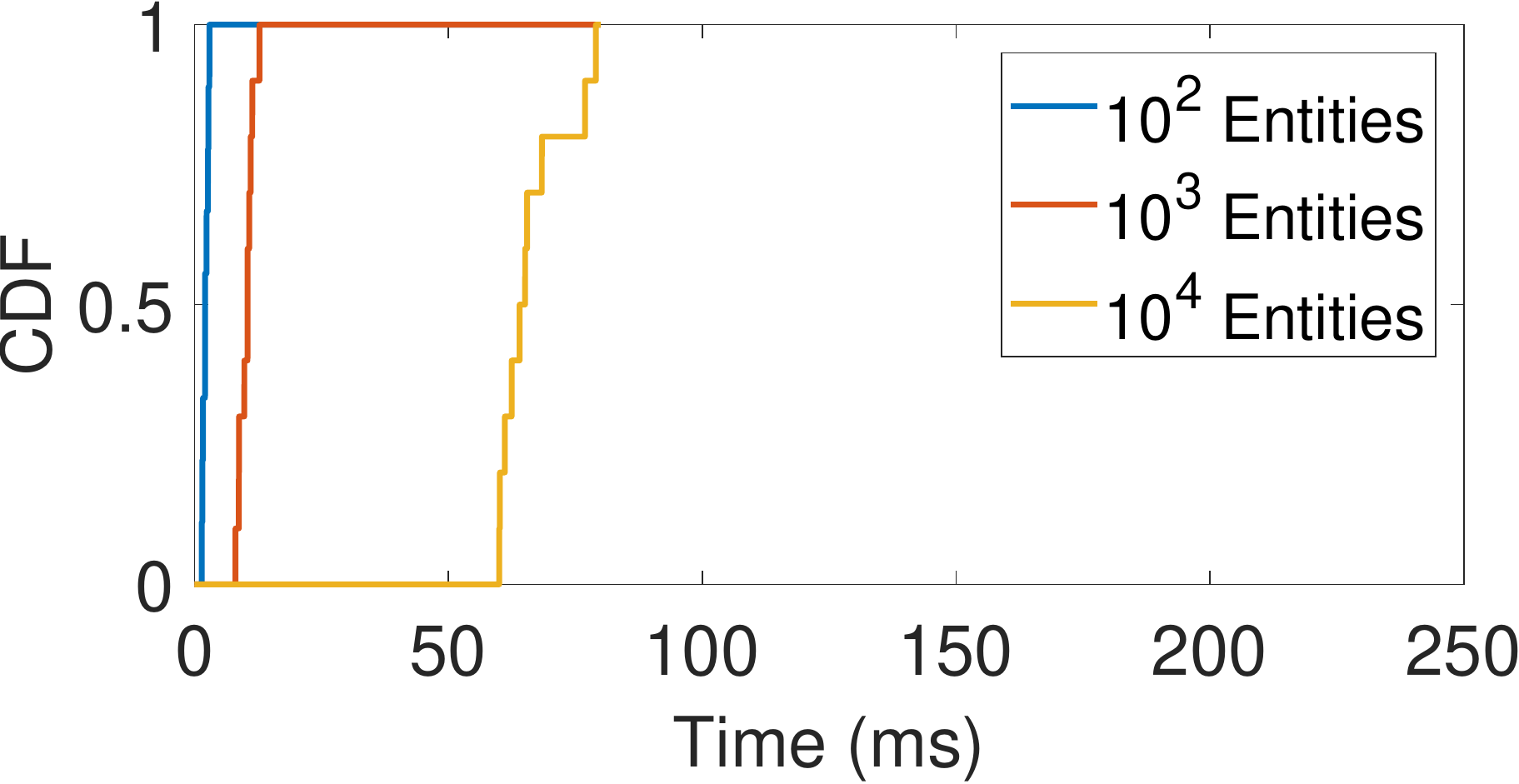}
}
\caption{The execution time for addition and multiplication operations on two vectors with $10^2$, $10^3$, $10^4$ entities.}
\vspace{-10pt}
\label{fig:additiveshare}
\end{figure}

\noindent{\bf Communication.} We measure the inter-cluster communication overhead, because it includes the main communication overhead in \system, which is to send the garbled sorting circuit as well as the labels of inputs to the counter-party.
Note that this is much larger than sending the multiplication triplets ($12$ bytes for each) and the encrypted $\id$ ($16$ bytes for each).
We demonstrate the communication overhead for different size of the circuit in Table~\ref{tlb:gc}. It shows that for an average user with approximately $130$ friends, the two-party only requires to transmits $9.80$MB data to sort them. 
This overhead is negligible both in our evaluation platform (40Gbps NIC in Azure intranet) and the other public clouds such as AWS.

\noindent{\bf Throughput.} To evaluate the impact of our system on throughput, we measure the server throughput for different types of operators. For each operator, we compare the throughput results between \system~and the baseline. Figure~\ref{fig:throughput} and Table~\ref{tlb:globalthroughput} show the throughput test results for \system~and baseline. 
In all the cases, we group our 6 VM instances into 2 clusters with 3 VMs to fulfil the two-party settings. All VMs are running with only one core involved in the computation, and we simulate $10000$ parallel client processes to send the query to the server, which ensures $100\%$ workload on the server side.
The results show that the throughput penalty is mainly from the sorting: the local sorting decreases the throughput by $38\%$ to $49\%$. We also observe that the global sorting is the bottleneck of the whole system (see Table~\ref{tlb:globalthroughput}), it gives a constant query throughput for all operators, which means it runs longer to obtain the final result. However, for the operators without sorting, the throughput loss is modest (only $4\%$ to $16\%$).

\begin{table}[!t]
\centering
	\caption{Throughput (Queries/sec) comparison of global sorting for query with different operators.}
	\label{tlb:globalthroughput}
	\begin{tabular}{|c|c|c|c|}
	\hline
	Operators & \textbf{term} & \textbf{and/diff.} & \textbf{or} \\
	\hline
	Baseline & $350$ & $321$ & $301$\\
	\hline
	Our system & $80$ &$80$& $80$\\
	\hline
	\end{tabular}
	\vspace{-10pt}
\end{table}
\begin{figure}[!t]
\includegraphics[width=\linewidth]{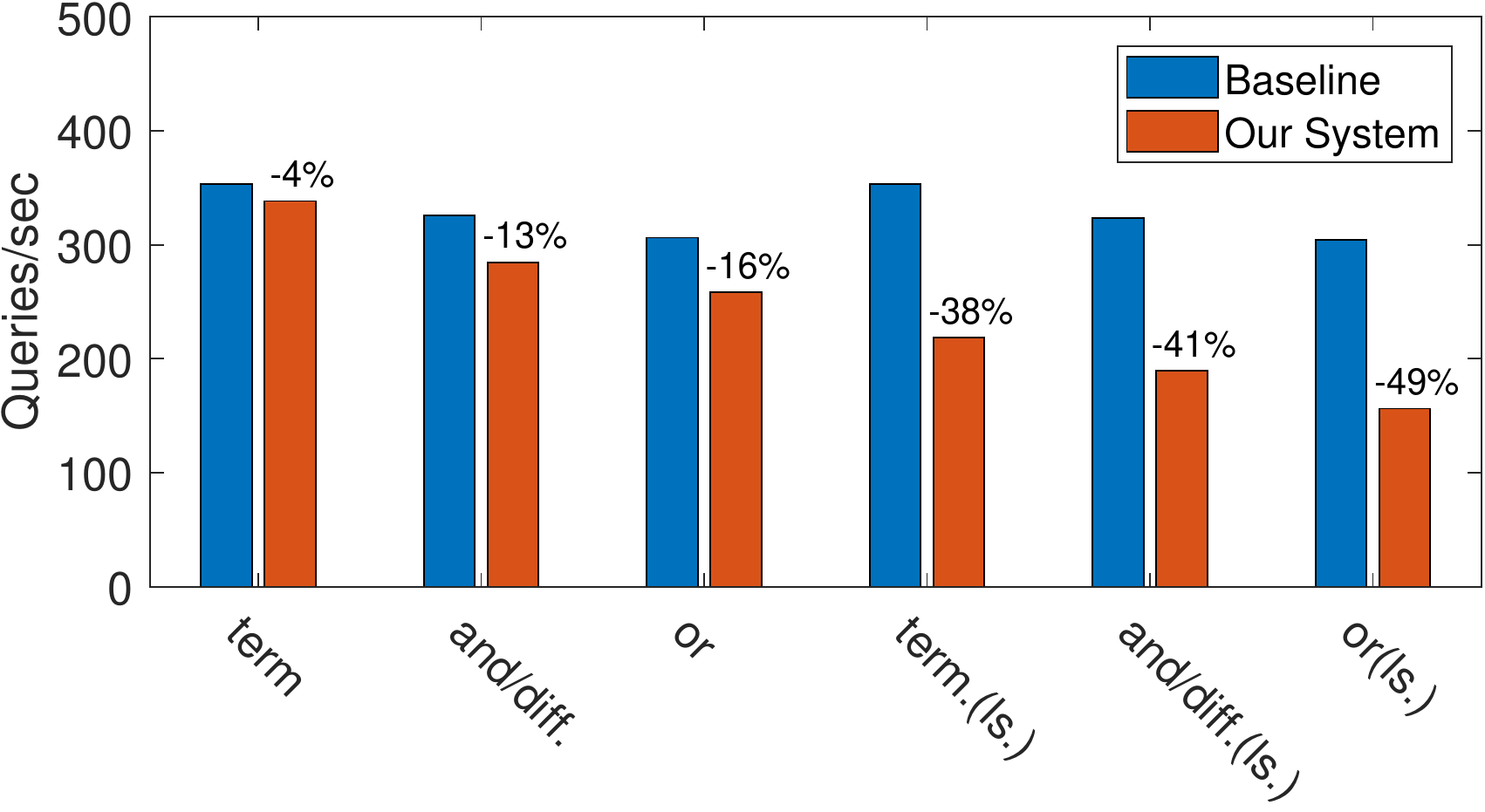}
\caption{Throughput of different types of Query Operators with $10000$ concurrent clients running under \system~and baseline, all operators except \textbf{term} have two keywords. Diff. stands for \textbf{difference} operator; ls. stands for locally sorted in $\IS$ after applying the operators.}
\vspace{-10pt}
\label{fig:throughput}
\end{figure}

\section{Conclusion}\label{sec:conclusion}
This paper presents an encrypted graph database, named \system. It enables privacy-preserving rich queries in the context of social network services. Our system leverages the advanced cryptographic primitives (i.e., $\OXT$, and mixing protocol with additive sharing and Garbled Circuit) with strong security guarantees for queries on structured social graph data, and queries with computation, respectively. To lead to a practical performance, \system~generates an encrypted index on a distributed graph model to facilitate parallel processing of the proposed graph queries. \system~is implemented as a prototype system, and our evaluation on YouTube dataset illustrates its efficiency on social search.

\bibliographystyle{ACM-Reference-Format}
\bibliography{reference}

\appendix

\section{Security Definition and Proof}
\subsection{Leakage Function}\label{subsec:leak}
Recall the security definition from~\cite{curtmola2011searchable,cash2013highly}: The security of SSE is parameterised by a leakage function $\Leakage_{\bf SSE}$, which depicts the scope of information about data and queries that the adversary is allowed to learn through the interaction with server.

Therefore, we start by giving the leakage function of query operators in \system.
As shown in Section~\ref{subsec:atomic}, all operators in \system~inherit the leakage of $\OXT$~\cite{cash2013highly} (i.e., $\TSet$ leakage $\Leakage_{\bf T}$ and $\OXT$ leakage $\Leakage_{\OXT}$).

Furthermore, all operators produce a sorted list as the return, it additionally introduces a new leakage about the rank of retrieved entities. We define a new leakage function $\Leakage_{\rm R}$ to capture this new leakage.
In particular, $\Leakage_{\rm R}$ consists of two sub-functions $\{f_0, f_1\}$ which are defined as follows:
\begin{itemize}[leftmargin=*]
\setlength{\itemsep}{0pt}
\setlength{\parsep}{0pt}
\setlength{\parskip}{0pt}
\item $f_0$: it takes as input the transcript of $\OXT$ and outputs a sorting circuit $F_{\it sort}$ and its input labels.
\item $f_1$: it takes as input the transcript of $\OXT$ and outputs the rank $r_{\id}$ for every $e_{\id}$, where $r_{\id}\in\mathcal{N}$ is the rank of $e_{\id}$.
\end{itemize}
Note that the adversary can only corrupt one of the two parties in our system which means the adversary only can access one of the above two sub-functions during the simulation.

The overall leakage function $\Leakage$ consists of the leakage from $\OXT$ as well as $\Leakage_{\rm R}$.

\subsection{Security Proofs}\label{subsec:proof}
Based on the above leakage function, we give a security analysis for \system~following the real/ideal paradigm. 

To start with, we define the real/ideal models for the query operators (e.g., {\bf term}, {\bf and}, etc.) that only involve structured data access and Garbled Circuit sorting, we denote the execution of the above query operators as query protocol $\Pi_1$.

The query protocol $\Pi_1$ is executed between a client $\Client$ ($\SF$) and two parties $\Party_i, i\in\{0,1\}$ ($\ISC$s).  
In real model, an adversary $\Adv$ can choose a social dataset $\DB$ and let $\Client$ generate the corresponding encrypted database $\EDB$ and give it to $\Adv$.
Then, $\Adv$ chooses a query list $\query$ to run in the two-party servers. To respond, each party firstly accesses $\EDB$ to retrieve contents and gives the transcript to $\Adv$.
Later, they perform two-party sorting via Garbled Circuit scheme. During the sorting, $\Adv$ is able to see all inbound/outbound messages in the corrupted party.

We let $\View^i_{\Pi_1}(1^{\lambda}, \DB, \query)$ be the entire view of $\Party_i$ in an execution of $\Pi_1$. Let $\Out^i_{\Pi_1}(1^{\lambda}, \DB, \query)$ be the outputs of party $\Party_i$ in the end of a $\Pi_1$ execution. 
Considering the security assumption of untrusted but non-colluded adversaries on the two-party setting, an adversary $\Adv$ can corrupt one party at most. The assumption restricts that $\Adv$ only can get the entire view of the corrupted party and the outputs from the counter-party.
Hence, $\View^i_{\Pi_1}(1^{\lambda}, \DB, \query)$ and $\Out^{(1-i)}_{\Pi_1}(1^{\lambda}, \DB, \query)$ are exactly the real model of $\Adv$ who corrupts $\Party_i$. In this case, we denote the adversary as $\Adv_i$ and set \\
$\Real_{\Pi_1, \Adv_i}(1^{\lambda}, \DB, \query)\overset{def}{=}$ \\
\hspace*{\fill}$(\View^i_{\Pi_1}(1^{\lambda}, \DB, \query), \Out^{(1-i)}_{\Pi_1}(1^{\lambda}, \DB, \query))$.

In the ideal model, the $\EDB$ of the chosen $\DB$ is generated by the simulator of $\OXT$ $\Sim_{\OXT}$. 
In the query phase, each party processes the chosen query list $\query$ and returns the transcript to $\Adv$ via $\Sim_{\OXT}$. Then, they hand their inputs to a trusted party $\mathit{TP}$ to perform sorting. 
We denote the above protocol executed in the ideal model as $\Pi_1^{'}$.
The view of $\Adv_i$ in the ideal model consists of the view on $\Party_i$ and the output of counter-party $\Party_{(1-i)}$. We set \\
$\Ideal_{\Pi_1^{'}, \Adv_i}(1^{\lambda}, \DB, \query)\overset{def}{=}$\\ 
\hspace*{\fill}$(\View^i_{\Pi_1^{'}}(1^{\lambda}, \DB, \query), \Out^{(1-i)}_{\Pi_1^{'}}(1^{\lambda}, \DB, \query))$.

Before we formalise the security of $\Pi_1$, we give the definition of {\it computationally indistinguishable}:
\begin{Definition}[Computationally Indistinguishable]
Assuming a distribution ensemble $\mathcal{X}=\{\mathcal{X}_i\}_{i\in\mathcal{I}}$ is a sequence of random variables indexed by $\mathcal{I}$.
Two distribution ensembles $\mathcal{X}=\{\mathcal{X}_i\}_{i\in\mathcal{I}}$ and $\mathcal{Y}=\{\mathcal{Y}_i\}_{i\in\mathcal{I}}$ are {\it computationally indistinguishable}, denoted as $\mathcal{X}\overset{c}{\equiv}\mathcal{Y}$, if for every probabilistic polynomial-time distinguisher $\mathcal{D}$, there exists a negligible function $\mathsf{negl}(\cdot)$, such that for every $i\in\mathcal{I}$\\
\vspace{-5pt}
$$|\mathrm{Pr}[\mathcal{D}(\mathcal{X}_i)=1] - \mathrm{Pr}[\mathcal{D}(\mathcal{Y}_i)=1]|\leq \mathsf{negl}(\lambda)$$.
\vspace{-15pt}
\end{Definition}

The security of $\Pi_1$ is defined as follows:
\begin{Definition}
Let $\Pi_1$, $\Pi_1^{'}$ be as above. protocol $\Pi_1$ is said to be $\Leakage$-semantically secure where $\Leakage$ is the leakage function defined as before if for every non-adaptive adversary $\Adv_i$ in the real model, there exists a simulator $\mathcal{S}_i$ in the ideal model such that \\
$\{\Real_{\Pi_1, \Adv_i}(1^{\lambda}, \DB, \query)\}_{i\in\{0,1\}}\overset{c}{\equiv}\{\Ideal_{\Pi_1^{'}, \mathcal{S}_i}(1^{\lambda}, \DB, \query)\}_{i\in\{0,1\}}$
\end{Definition}

We further define our adaptive model as in~\cite{cash2013highly}. In such a model, the query list $\query$ is not given to the challenger in the above two games. Instead, $\Adv$ adaptively chooses each query after receiving $\EDB$.

We assume that $\Pi_1$ runs in a {\it hybrid model} where parties are given access to the trusted party computing the ideal function of OT.
We show that $\Pi_1$ is secure with the given leakage function in this hybrid model. It follows from the standard composition theorems~\cite{canetti2000security} that $\Pi_1$ is secure with the given leakage function if the trusted party is replaced by secure protocol (i.e., real OT).
\begin{Theorem}
\label{thm:non-adaptive}
In the OT-hybrid model execution, protocol $\Pi_1$ is $\Leakage$-semantically secure against non-adaptive adversaries, assuming that $\OXT$ protocol is $\Leakage_{\OXT}$-semantically secure against non-adaptive adversaries, that the Garbled Circuit (GC) construction is secure against semi-honest but non-colluding adversaries.
\end{Theorem}

\begin{proof}
Let $\Adv_i$ denote an adversary corrupts $\Party_i$ attacking the protocol in a hybrid model where the parties run OT via trusted entities (OT-hybrid model).
As the view of garbler and evaluator are different in $\Pi_1$, we consider separately the cases $i = 0$ and $i = 1$. We further denote the ensemble of $\Client$, the counter-party of $\Party_i$ and the trusted party $\mathit{TP}$ as the challenger $\Clg$ in our simulation.
In both cases, we show that we can construct a simulator $\mathcal{S}_i$ running in our ideal model where the parties involve a trusted entity computing $\Pi_1^{'}$, which returns the same view as $\Adv_i$'s in the hybrid model.

\noindent$\mathcal{S}_i$ is constructed as follows:
\begin{enumerate}[leftmargin=*]
\setlength{\itemsep}{0pt}
\setlength{\parsep}{0pt}
\setlength{\parskip}{0pt}
\item $\mathcal{S}_i$ chooses the $\DB$ and $\query$ and simulates $\Adv$. 
It firstly gives $\DB$ to $\Clg$. $\Clg$ runs $\Sim_{\OXT}(\Leakage_{\OXT}(\DB))$ and return $\EDB$ to $\mathcal{S}_i$.

\item For the given query list $\query$, $\mathcal{S}_i$ processes $\Sim_{\OXT}(\Leakage_{\OXT}(\DB, q))$ and receives the output $\mathbf{tr}_i$ from $\Sim_{\OXT}$. 
The output $\mathbf{tr}_i$ is identical to the transcript of $\OXT$ except that the query result of $\OXT$ only includes the encrypted id $e_{\id}$, but in $\Pi_1^{'}$, the query result is a list of encrypted tuples $\{\mathbf{E}\}=\{(\langle x\rangle^A, e_{\id})\}$.

This completes the the simulation of the structured data access part of $\Pi_1$.

\item Next, $\mathcal{S}_i$ hands $\mathbf{tr}_i$ to $\Clg$ and simulates the sorting procedure. As each party runs the same query on the same partition of the $\OXT$ index, the size of transcript is equal. $\Clg$ runs $f_0(\mathbf{tr}_0,\mathbf{tr}_1)$ to generate an array of sorting circuits $\mathbf{F}$ with $|\mathbf{tr}|$ circuits and $f_1(\mathbf{tr}_0,\mathbf{tr}_1)$ to generate an array of ranks $\mathbf{R}[e_\id]$ indexing by the $e_\id$ from $\mathbf{tr}$.

\item For each $\mathbf{tr}_i[j], 1\leq j\leq |\mathbf{tr}|$, if $i=0$ (garbler):
\begin{enumerate}[leftmargin=*]
\setlength{\itemsep}{0pt}
\setlength{\parsep}{0pt}
\setlength{\parskip}{0pt}
\item $\mathcal{S}_0$ sends $\mathbf{tr}_0[j]$ to the $\Clg$ and receives $\mathbf{F}[j]$ and its input labels $\langle\mathrm{X}\rangle_0^Y$ from $\Clg$. In addition, $\mathcal{S}_0$ is able to get the labels $\langle\mathrm{X}_0\rangle_1^Y$ of its additive shares $\langle \mathrm{X}\rangle_0^A$. There is no message from the counter-party. 

\item The counter-party ($\Party_1$) uses its $\mathbf{tr}_1[j]$ to retrieve a $\Leakage_{\rm R}$-simulated ranked list $\{\bar{\mathbf{E}}\}=\{(r_{\id} ,e_{\id})\}$ from $\mathbf{R}[e_\id]$.

\item After execution, $\mathcal{S}_0$ outputs $\{\EDB, \mathbf{tr}_0[j], \mathbf{F}[j], \langle\mathrm{X}\rangle_0^Y, \langle\mathrm{X}_0\rangle_1^Y\}$ as $\View^0_{\Pi_1^{'}}(1^{\lambda}, \DB, \query)$ and $\{\{\bar{\mathbf{E}}\}\}$ as $\Out^1_{\Pi_1^{'}}(1^{\lambda}, \DB, \query)$.
\end{enumerate}

\item If $i=1$ (evaluator):
\begin{enumerate}[leftmargin=*]
\setlength{\itemsep}{0pt}
\setlength{\parsep}{0pt}
\setlength{\parskip}{0pt}
\item $\mathcal{S}_1$ simulates OT protocol by generating the labels $\langle\mathrm{X}_1\rangle_1^Y$ for each evaluator's input $\langle x\rangle_1^A$.

\item $\mathcal{S}_1$ retrieves ranks from $\mathbf{R}[e_\id]$ and constructs $\{\bar{\mathbf{E}}\}$.

\item Then, in the usual way (e.g.~\cite{lindell2009proof, bellare2012foundations}), $\mathcal{S}_1$ simulates a sorting circuit $F_{\it sort}$ in such a way that given the input labels chosen by $\mathcal{S}_1$, the output of the circuit is precisely $\{\bar{\mathbf{E}}\}$.

\item Finally, $\mathcal{S}_1$ outputs $\{\EDB, \mathbf{tr}_1[j], F_{sort}, \langle\mathrm{X}_0\rangle_1^Y, \langle\mathrm{X}_1\rangle_1^Y, \{\bar{\mathbf{E}}\}\}$ as $\View^1_{\Pi_1^{'}}(1^{\lambda}, \DB, \query)$ and $\{\emptyset\}$ as $\Out^0_{\Pi_1^{'}}(1^{\lambda}, \DB, \query)$.
\end{enumerate}
\end{enumerate}

To complete the proof, we use the above simulator to show that\\
$\{\Real_{\Pi_1, \Adv_i}(1^{\lambda}, \DB, \query)\}_{i\in\{0,1\}}\overset{c}{\equiv}\{\Ideal_{\Pi_1^{'}, \mathcal{S}_i}(1^{\lambda}, \DB, \query)\}_{i\in\{0,1\}}$:

\begin{enumerate}[leftmargin=*]
\setlength{\itemsep}{0pt}
\setlength{\parsep}{0pt}
\setlength{\parskip}{0pt}
\item $\EDB$ and $\mathbf{tr}$ in the real model are generated by running $\OXT$ protocol, while in the ideal model, they are simulated by given the output of leakage function.
\item For $i=0$ (garbler):
\begin{enumerate}[leftmargin=*]
\setlength{\itemsep}{0pt}
\setlength{\parsep}{0pt}
\setlength{\parskip}{0pt}
\item In the real model, the garbler constructs the sorting circuit correctly based on garbler's inputs, while in the ideal model the garbled circuit is a dummy circuit in $\mathbf{F}$ with the given input labels from $\Clg$.
\item In the real model, the evaluator gets the output $\{\bar{\mathbf{E}}\}$ after evaluating the circuit. In the ideal model, the output is
obtained from the trusted entity via the output of leakage function $\Leakage_{\rm R}$.
\end{enumerate}

\item For $i=1$ (evaluator):
\begin{enumerate}[leftmargin=*]
\setlength{\itemsep}{0pt}
\setlength{\parsep}{0pt}
\setlength{\parskip}{0pt}
\item In the real model, $F_{\it sort}$ and the input label sent by the garbler is correctly generated based on garbler's inputs, while in the ideal model the garbled circuit is simulated based on the random generated input labels and the output of leakage function $\Leakage_{\rm R}$ from $\Clg$. It ensures that the circuit returns the same output in two models.
\item Both in the real and ideal model, the garbler does not output a result in the end of the protocol execution.
\end{enumerate}
\end{enumerate}

The security properties of $\OXT$ ensures the indistinguishability of structured data access part. 
For the sorting part, The privacy property of Garbled Circuit~\cite{bellare2012foundations} ensures that the adversary only can learn the inputs from the circuit and output with a negligible probability.
It concludes that for every non-adaptive adversary $\Adv_i$, it has negligible probability to learn more information than the defined leakage function $\Leakage$. 
\end{proof}

We now show that our theorem is also valid for adaptive models.
\begin{Theorem}
\label{thm:adaptive}
In the OT-hybrid model execution, protocol $\Pi_1$ is $\Leakage$-semantically secure against adaptive adversaries, assuming that $\OXT$ protocol is $\Leakage_{\OXT}$-semantically secure against adaptive adversaries, that the Garbled Circuit (GC) construction is secure against semi-honest but non-colluding adversaries.
\end{Theorem}

\begin{proof}
First of all, $\OXT$ protocol has been proved to be $\Leakage_{\OXT}$-semantically secure against adaptive adversaries~\cite{cash2013highly}. The only concern here is how to handle the adaptivity towards sorting.

To simulate the response for adaptive queries, the simulator should adaptively generate the ranks $\mathbf{R}[e_\id]$ for $e_\id$ list, This is in contrast to the non-adaptive simulator, where it can generates $\mathbf{R}[e_\id]$ as determine by the leakage. 
Instead, the simulator uses the STag and $e_\id$ list from the transcript to maintain a bidimensional array $\mathbf{R}[STag, e_\id]$. It adaptively updates $\mathbf{R}[STag, e_\id]$ for each new $e_\id$ which does not exist in $\mathbf{R}[STag, e_\id]$.
\end{proof}

Finally, we show that the defined operator (i.e., {\bf apply}) is also secure after involving arithmetic computations. We slightly modify our hybrid model by adding an ideal function $f_A$, which evaluates arbitrary addition/multiplications using the ideal function of additive sharing scheme. The new hybrid model is defined as ($f_A$,OT)-hybrid model.
 \begin{Theorem}
In the ($f_A$,OT)-hybrid model execution, protocol $\Pi_1$ is $\Leakage$-semantically secure against adaptive adversaries, assuming that $\OXT$ protocol is $\Leakage_{\OXT}$-semantically secure against adaptive adversaries, that the Garbled Circuit (GC) construction is secure against semi-honest but non-colluding adversaries.
\end{Theorem}

\begin{proof}
We have shown that the output ranks can be simulated without using the actual additive shares in Theorem~\ref{thm:non-adaptive},~\ref{thm:adaptive}. 
Informally, it also indicates that $\Pi_1$ can also securely perform structured data access and sorting under the ($f_A$,OT)-hybrid model.
In addition, previous work shows that the additive sharing scheme can securely compute the given arithmetic formulas~\cite{pullonen2012design}.
The standard sequential modular composition~\cite{canetti2000security} implies that $\Pi_1$ with sub-protocols evaluating $f_A$ and OT remains secure with the same leakage $\Leakage$ in the real model.
\end{proof}

\balance

\end{document}